\documentclass[a4paper,11pt,twoside]{article}

\pdfoutput=1

\makeatletter

\usepackage[english]{babel}
\usepackage{cancel,graphicx}
\usepackage[pdftex]{hyperref}
\usepackage{amsmath,amssymb}
\usepackage{chngcntr}
\usepackage{cancel,graphicx}
\usepackage{hyperref}
\usepackage{amsmath,amssymb,slashed}
\usepackage{amsmath,amssymb}
\usepackage[amsmath,hyperref,thmmarks]{ntheorem}
\usepackage{geometry}
\usepackage{color}
\usepackage{chngcntr}
\usepackage{url}

\allowdisplaybreaks[1]
\numberwithin{equation}{section}

\definecolor{hypercolor}{rgb}{0.1,0.2,0.6}

\hypersetup{     
unicode=false,      
pdftoolbar=true,    
pdfmenubar=true,    
pdffitwindow=true,  
pdfstartview={FitH},
pdftitle={Draft},    
pdfauthor={Antoine G\'er\'e, Jean-Christophe Wallet},     
pdfsubject={Mathematical Physics},   
pdfcreator={LaTeX},  
pdfproducer={pdfTex},
pdfkeywords={Noncommutative Field Theory},  
pdfnewwindow=true,  
colorlinks=true, 
linkcolor=hypercolor, 
urlcolor=hypercolor, 
citecolor=hypercolor,
filecolor=hypercolor,         
} 

\newcommand{\tr}{\mathsf{tr}}

\newcommand{\eqn}[1]{(\ref{#1})}

\newcommand{\be}{\begin{equation}}
\newcommand{\ee}{\end{equation}}

\newcommand\bbone{{ \mathbb{I}}}

\theoremclass{LaTeX}

\theoremsymbol{}
\theoremseparator{}
\theorembodyfont{\slshape}
\theoremheaderfont{\normalfont\bfseries}

\newtheorem{Theorem}{Theorem}[section]
\newtheorem{theorem}[Theorem]{Theorem}
\newtheorem{atheorem}{Theorem}

\newtheorem{proposition}[Theorem]{Proposition}

\newtheorem{lemma}[Theorem]{Lemma}

\theorembodyfont{\upshape}
\theoremsymbol{\ensuremath{\blacklozenge}}

\newtheorem{remark}[Theorem]{Remark}

\newtheorem{definition}[Theorem]{Definition}
\newtheorem{adefinition}{Definition}

\theoremstyle{nonumberplain}
\theoremheaderfont{\scshape}
\theorembodyfont{\normalfont}
\theoremsymbol{\ensuremath{\blacksquare}}

\newtheorem{proof}{Proof}

\qedsymbol{\ensuremath{_\blacksquare}}

\newcommand{\institute}[1]{\newcommand{\@institute}{#1}}

\renewcommand{\maketitle}{
\vspace*{0.5\baselineskip}
{
\center\LARGE\noindent\@title\par
}%
\vspace{1.5\baselineskip}
{
\center\normalsize\noindent\ignorespaces\@author\par
}%
\vspace{0.5\baselineskip}
{
\center\normalsize\ignorespaces\@institute\par
}%
\vspace{2\baselineskip}
}%

\begin{document}

\title{Spectral theorem in noncommutative field theories: Jacobi dynamics}

\author{Antoine G\'er\'e$^a$ and Jean-Christophe Wallet$^c$}

\institute{%
\textit{$^a$Dipartimento di Matematica, Universit\`a di Genova\\
Via Dodecaneso, 35, I-16146 Genova, Italy}\\
e-mail:\href{mailto:gere@dima.unige.it}{\texttt{gere@dima.unige.it}}\\[1ex]%
\textit{$^c$Laboratoire de Physique Th\'eorique, B\^at.\ 210\\
CNRS and Universit\'e Paris-Sud 11,  91405 Orsay Cedex, France}\\
e-mail:\href{mailto:jean-christophe.wallet@th.u-psud.fr}{\texttt{jean-christophe.wallet@th.u-psud.fr}}
}%

\maketitle


{{
\small
%


}}


\begin{abstract}
Jacobi operators appear as kinetic operators of several classes of noncommutative field theories (NCFT) considered recently. This paper deals with the case of bounded Jacobi operators. A set of tools mainly issued from operator and spectral theory is given in a way applicable to the study of NCFT. As an illustration, this is applied to a gauge-fixed version of the induced gauge theory on the Moyal plane expanded around a symmetric vacuum. The characterization of the spectrum of the kinetic operator is given, showing a behavior somewhat similar to a massless theory. An attempt to characterize the noncommutative geometry related to the gauge fixed action is presented. Using a Dirac operator obtained from the kinetic operator, it is shown that one can construct an even, regular, weakly real spectral triple. This spectral triple does not define a noncommutative metric space for the Connes spectral distance.
\end{abstract}

\section{Introduction}

Noncommutative Geometry (NCG) \cite{Connes1} (see also \cite{GBVF}) may provide a way to escape physical obstructions to the existence of continuous space-time and commuting coordinates at the Planck scale \cite{Doplich1}. This has reinforced the interest in noncommutative field theories (NCFT), which appeared in their modern formulation first in String field theory \cite{witt1}, followed by models on the fuzzy sphere and almost commutative geometries \cite{mdv1}, \cite{gm90} while NCFT on noncommutative Moyal spaces received attention from the end of the 90's. For reviews, see for instance \cite{dnsw-rev}.\par 

It appears that the kinetic operators for the actions of several classes of NCFT considered so far are Jacobi operators \cite{akhiez:1965}. Jacobi operators can be presented informally as symmetric tridiagonal infinite matrices related to a 3-term recurrence (defining) relation \cite{akhiez:1965}. It turns out that Jacobi kinetic operators manifest themselves under this matrix form when the action is expressed within a so called matrix base formalism. Of course, not all NCFT have Jacobi kinetic operators but this latter case encompasses several NCFT of interest. This is for instance the case of the Grosse-Wulkenhaar model on $\mathbb{R}^4_\theta$ and $\mathbb{R}^2_\theta$ \cite{gw1}, \cite{gw2}, its rotationally invariant descendants \cite{gwm-4} and fermionic extensions \cite{VT}, the scalar $\phi^4$ model on $\mathbb{R}^3_\lambda$ \cite{vit-wal-12} together with the massless gauge theory on $\mathbb{R}^3_\lambda$ \cite{gvw13} whose truncation to the fuzzy sphere is related to the Alekseev-Recknagel-
Schomerus matrix model \cite{ARS}.\par 

To a given Jacobi operator corresponds a given family of orthogonal polynomials. This is the actual reason why definite families of orthogonal polynomials are automatically singled out in the diagonalization procedure of the kinetic operators for the above mentionned NCFT, and exemplified in this paper for another specific NCFT. Orthogonal polynomials are subject of intense activity in applied mathematics, ranking from approximation theory and numerical calculus to constructions of new families of multivariable orthogonal polynomials and families of random variables. The NCFT considered in this paper corresponds to Chebyshev polynomial of 2d kind, a particular type of orthogonal Jacobi polynomials. These determine entirely the kinetic operator via the spectral theorem, hence the (free) dynamics of the system. This explains the statement ''Jacobi dynamics'' in the title of the paper.\par 

The numerous families of orthogonal polynomials are classified within the Askey scheme that will not be essential in the present paper. For more details, see e.g \cite{kks} and references therein. The fact that an orthogonality relation can be defined within the above orthogonal families stems from the existence of an integration measure which is a mere consequence of the spectral theorem, sometimes put in a particular form known as the Favard theorem \cite{favard} which may be traced back to works of Stieltjes \cite{stieltjes}. The measure can be determined from its moments, giving rise to the so called (determinate or indeterminate) moment problem, that will not be central in the present paper. Note that the notion of moments is familiar to the theory of probability. There is a huge mathematical literature related to these topics, among which the reader is invited to see e.g \cite{akhiez:1965}, \cite{Gszego}, \cite{bsimon}, \cite{teschl}.\par 

In this paper, we consider the case of bounded Jacobi operators. We select among the different mathematical domains mentionned above the mathematical tools relevant to the construction and the study of (quantum) NCFT, focusing on functional analysis formalism which is the most suitable one to deal with theories expressed within ''matrix bases''. The corresponding mathematical toolkit tailored for a use in NCFT is presented in the Section \ref{subsection21}. Note that many of the properties given in Section \ref{subsection21} still apply to the case of unbounded Jacobi operators. This latter case will be presented elsewhere \cite{unboud-jac}. \\
As a representative example, we apply the material of Section 2 to a NCFT on the Moyal plane $\mathbb{R}^2_\theta$. This is presented in the Section \ref{subsection22}. The NCFT is a gauge-fixed version of the so called ''induced gauge theory action'' \cite{GWW}, \cite{GW07} (see \eqref{inducedgaugematrix} in Section 3), expanded around a particular symmetric vacuum \cite{GWW2} which has been considered recently \cite{MVW13}. The NCFT has a bounded Jacobi kinetic operator and polynomial interactions after suitable gauge fixing. A complete characterization of the spectrum of the kinetic operator is given together with a very simple derivation of the propagator, stemming directly from the spectral theorem. A review on the situation of noncommutative gauge theories is given in subsection \ref{section-review}.\\
The induced gauge theory action \eqref{inducedgaugematrix} can be related to a particular NCG described by a finite volume spectral triple \cite{finite-vol}, homothetic as a noncommutative metric space \cite{homot-moyal} to the standard noncommutative Moyal space \cite{marseil1} whose Connes spectral distance has been investigated in \cite{moyal1}. Note that a theory of quantum (i.e noncommutative) compact metric spaces has been set up in \cite{Rieffel}. An extension to the (noncommutative analog of) locally compact case has been given in \cite{latrem}, to which Moyal spaces belong. However, the NCFT of section \ref{subsection22} results from \eqref{inducedgaugematrix} after an expansion around some vacuum and a gauge-fixing which thus modify the initial kinetic operator in \eqref{inducedgaugematrix} and so the related Dirac operator. One question would be to characterize the NCG related to the resulting NCFT which would require to built a spectral triple and then perform the full computation of some
spectral action with relevant perturbation of the Dirac operator. The Section \ref{ncg} is a partial attempt in that direction. We construct spectral triples from a Dirac operator whose square equals the bounded Jacobi kinetic operator. In the section \ref{discuss}, we discuss the results and conclude.\par 
We find a complete characterization of the spectrum of the kinetic operator. In particular, its spectrum involves $0$ as a limit point which may be interpreted as a massless kinetic operator. The corresponding unbounded propagator is obtained as a simple application of Section \ref{subsection21}, therefore completing the first analysis performed in \cite{MVW13}.  At the perturbative level, the NCFT exhibits unsuppressed correlations at large separation between indices. By using a Dirac operator built from the kinetic operator, we show that one can construct a spectral triple. This spectral triple cannot defines a noncommutative metric space for the Connes spectral distance, i.e is not a spectral metric space in the sense of \cite{homot-moyal}, \cite{bel-mar}.


\section{Toolkit for NCFT: Spectral theorem and Jacobi operators}\label{subsection21}

This section is written for physicists. We collect the necessary (basic) mathematical background needed to characterize spectral properties underlying the actions of several classes of NCFT as well as matrix models recently considered. The material is based on Jacobi operators, orthogonal polynomials together with a special version of the spectral theorem, sometimes called the Favard theorem. General definitions and usefull properties of operator algebra are collected in \ref{appendix1} to make the presentation self contained. In particular, we explicitely state in \ref{appendix1} the special version of the spectral theorem that is at the crossroad of the Jacobi operators, orthogonal polynomials used here and spectral properties of kinetic operators of the NCFT considered in the paper. For more details, the reader should borrow material from e.g \cite{akhiez:1965}, \cite{Gszego}, \cite{bsimon}, \cite{teschl}.\par
Let $\ell^2(\mathbb{N})$ be the Hilbert spaces of square integrable sequences $(u_n)_{n\in\mathbb{N}}$ with canonical Hilbert product $\langle (u_n),(v_m)\rangle:=\sum_{k\in\mathbb{N}}u_kv_k^*$. Let $\mu$ be a positive Borel measure on $\mathbb{R}$ satisfying $\int_\mathbb{R} x^nd\mu(x)<\infty$ for any $n\in\mathbb{N}$ and $\int_\mathbb{R} d\mu(x)=1$. Let $L^2(\mu)$ denotes the Hilbert space{\footnote{It is understood that 2 functions $a$ and $b$ are identified whenever $\int_\mathbb{R}\vert a(x)-b(x) \vert^2d\mu(x)=0$.}} of squared integrable functions on $\mathbb{R}$ with product $\langle a,b\rangle:=\int_\mathbb{R}a(x)b(x)^*d\mu(x)$. 
\begin{definition}\label{defin-jacobi}
A Jacobi operator $J$ acting on the Hilbert space $\ell^2(\mathbb{N})$ is defined by
\begin{eqnarray}
(Je)_m&=&a_me_{m+1}+b_me_m+a_{m-1}e_{m-1},\ m\ge1;\nonumber\\ 
(Je)_0&=&a_0e_1+b_0e_0 \label{jacoboperator},
\end{eqnarray}
where $\{e_m\}_{m\in\mathbb{N}}$ denotes the canonical orthonormal basis of $\ell^2(\mathbb{N})$ and $a_m\in\mathbb{R}^+$, $b_m\in\mathbb{R}$, $\forall m\in\mathbb{N}$.
\end{definition}
By using the explicit expression for the elements of the canonical basis, it is easy to represent \eqref{jacoboperator} as an infinite real symmetric tridiagonal matrix with $b_m$ (resp. $a_m$) $\forall m\in\mathbb{N}$ as diagonal (resp. upper and lower subdiagonals) elements. \par 
The properties of the Jacobi operator depend on the behavior of the sequences $\{a_m\}$, $\{b_m\}$, $m\in\mathbb{N}$. We start by listing simple useful properties:
\begin{proposition}\label{extend-self}
Let ${\cal{D}}(\ell^2(\mathbb{N}))$ be the set of finite linear combination of the elements of the canonical basis $\{e_m\}_{m\in\mathbb{N}}$ of $\ell^2(\mathbb{N})$.
\begin{itemize}
\vspace*{-3pt}
\setlength{\itemsep}{-1pt}
\item i) $J$ as defined in \eqref{jacoboperator} is a symmetric operator with dense domain ${\cal{D}}(\ell^2(\mathbb{N}))$.
\item ii) If $J$ is bounded, then it extends to a self-adjoint operator on $\ell^2(\mathbb{N})$.
\end{itemize}
\end{proposition}
\begin{proof}
By a simple calculation, one checks that $\langle Ju,v \rangle=\langle u,Jv \rangle$ for any $u,v\in{\cal{D}}(\ell^2(\mathbb{N}))$. Next, one observes that ${\cal{D}}(\ell^2(\mathbb{N}))$ is dense in $\ell^2(\mathbb{N})$. Hence, $J$ is a symmetric operator of $\ell^2(\mathbb{N})$ with dense domain ${\cal{D}}(\ell^2(\mathbb{N}))$ and i) is proven. Let $J^\dag$ the adjoint of $J$. When $J$ is bounded, the mere use of standard definitions yields $Dom(J)=Dom(J^\dag)=\ell^2(\mathbb{N})$ while $J$ and $J^\dag$ have the same action. Hence 
$J=J^\dag$ which proves ii).
\end{proof}
In the rest of this paper, we will consider only bounded Jacobi operators. Notice that a bounded Jacobi operator is automatically self-adjoint. We recall that not all symmetric densely defined operators have self-adjoint extensions. Their existence is controlled by the deficiency indices $(n_+,n_-)${\footnote{Let $T$ be a densely symmetric operator. The deficiency indices are defined as $n_\pm=\dim{\cal{N}}_{\pm i}$ where ${\cal{N}}_z:=\{u\in Dom(T^\dag):T^\dag u=zu \}$ for $z\in\mathbb{C}\backslash\mathbb{R}$}}. A necessary and sufficient condition for a densely symmetric operator to have self-adjoint extensions is that $n_+=n_-$ which will be automatically verified by our Jacobi operators in this paper since they are  assumed to be bounded, hence self-adjoint so that their spectrum is real, implying $n_+=n_-=0$.\par
For bounded Jacobi operators, there is a constraint on the related sequences $\{a_m\}$, $\{b_m\}$, $m\in\mathbb{N}$.
\begin{proposition}\label{jaco-bounded}
Let $J$ a Jacobi operator defined by the real sequences $\{a_m\}$, $\{b_m\}$, $m\in\mathbb{N}$ as in \eqref{jacoboperator}. The following properties hold true:
\begin{itemize}
\vspace*{-3pt}
\setlength{\itemsep}{-1pt}
\item i) If $J$ is bounded, then, $\{a_m\}_{m\in\mathbb{N}}$ and $\{b_m\}_{m\in\mathbb{N}}$ are bounded sequences.
\item ii) If $\{a_m\}_{m\in\mathbb{N}}$ and $\{b_m\}_{m\in\mathbb{N}}$ satisfy $\sup_{m\in\mathbb{N}}(|a_m|+|b_m|)\le M$, then $||J||\le2 M$.
\end{itemize}
\end{proposition}
\begin{proof}
Assume $J$ is bounded. Then, one has for any $u,v\in\cal{H}$ $|\langle Ju,v \rangle|\le ||J||$. This is true for $u=e_k$, $v=e_{k+1}$ and $u=e_k$, $v=e_k$ which using Definition \ref{defin-jacobi} yields $|a_m|\le||J||$, $b_m\le||J||$ which proves i). \\
Assume now $\sup_{m\in\mathbb{N}}(|a_m|+|b_m|)\le M$. For any unit vector $u\in\ell^2(\mathbb{N})$, $u=\sum_mu_me_m$, a simple calculation yields
\begin{eqnarray}
||Ju||^2&=&\sum_ma^2_{m-1}|u_{m+1}|^2+a^2_m|u_{m-1}|^2+b^2_m|u_m|^2+2a_mb_m(\Re(u_{m-1}u^*_m)+\Re(u_{m+1}u^*_m))\nonumber\\
&+&2a_ma_{m-1}\Re(u_{m-1}u^*_{m+1})\le\sum_ma^2(|u_{m+1}|^2+|u_{m-1}|^2)+b^2|u_m|^2\nonumber\\
&+&2ab(|\Re(u_{m-1}u^*_m)|+|\Re(u_{m+1}u^*_m)|)+2a^2|\Re(u_{m-1}u^*_{m+1})|\nonumber\\
&\le&2a^2+b^2+2ab(|\Re(u_{m-1}u^*_m)|+|\Re(u_{m+1}u^*_m)|)+2a^2|\Re(u_{m-1}u^*_{m+1})|\nonumber\\
&\le&2a^2+b^2+4ab|\langle Su,u \rangle|+2a^2|\langle S^2u,u\rangle|\le4a^2+b^2+4ab\le4M^2
\label{sum-ju-maj}
\end{eqnarray}
where we set $\sup_m|a_m|=a$, $\sup_m|b_m|=b$ and $S$ is the bounded shift operator defined by $
S:\ell^2(\mathbb{N})\to\ell^2(\mathbb{N}),\ S:e_m\mapsto e_{m+1},\ \forall m\in\mathbb{N}$
It satisfies obviously $||S||=1$ and $|\langle Su,u\rangle|\le1$, $|\langle S^2u,u \rangle|\le1$. Hence, $||J||\le2M$ and ii) is proven.
\end{proof} 
Now, apply the spectral theorem for bounded operator Theorem \ref{orthop-1} to $J$. This latter guarantees the existence of a natural (and unique) compactly supported measure rigidely linked to $J$. Indeed, from Theorem \ref{orthop-1}, there is a unique resolution $E$ of the identity on $\ell^2(\mathbb{N})$ such that for any $u,v\in\ell^2(\mathbb{N})$, 
\begin{equation}
\langle Ju,v \rangle=\int_\mathbb{R}tdE_{u,v}(t). 
\end{equation}
Let $U$ be a Borel set. From Definition \ref{resolution}, one can define a Borel measure by $U\mapsto E_{u,v}=\langle E(U)u,v \rangle$, for any $u,v\in\ell^2(\mathbb{N})$. Here, $e_0$ in the canonical basis of $\ell^2(\mathbb{N})$ is a convenient vector to be used by noticing in particular that the set $\{J^ne_0,\ n\in\mathbb{N}\}$ is dense in $\ell^2(\mathbb{N})$, as a standard computation shows, i.e $e_0$ is cyclic for the action of $J$. We have an immediate property summarized in the following proposition.
\begin{proposition}\label{measure-proba}
The measure $\mu$ defined by $\mu(U):=E_{e_0,e_0}=\langle E(U)e_0,e_0 \rangle$, for any Borel set $U$ is Borel positive.
\end{proposition}
\begin{proof}
Since $E(U)$ is an orthogonal projection, one has 
\begin{equation}
\mu(U)=\langle E(U)^2e_0,e_0  \rangle=\langle E(U)e_0,E(U)e_0\rangle\ge0. 
\end{equation}
Hence $\mu(U)$ is positive Borel with support $supp(\mu)\subseteq[-||J||,||J||]$. 
\end{proof}
The fact that $\mu$ can be related to polynomials shows up in the following proposition. It stems from the fact one can generate all the vectors of the canonical basis by a suitable polynomial action of $J$ on $e_0$.
\begin{proposition}\label{matrixelem-spectral-meas}
Let $J$ a bounded Jacobi operator with associated resolution of identity $E$ and $\{e_n\}_{n\in\mathbb{N}}$ the canonical basis of $\ell^2(\mathbb{N})$. For any $n\in\mathbb{N}$, one can find $p_n$ a polynomial of degree $n$ in $\mathbb{R}[X]$ such that
\begin{equation}
e_n=p_n(J)e_0.
\end{equation}
The spectral measure is completely determined by $\mu$, i.e
\begin{equation}
\langle E(U)e_n,e_m \rangle=\int_{U}p_n(t)p_m(t)d\mu(t).
\end{equation}
\end{proposition}
\begin{proof}
First, one checks by induction that $e_n=p_n(J)e_0$, $\forall n\in\mathbb{N}$ for some polynomials $p_n\in\mathbb{R}[X]$, thanks to $a_n>0$. \\
Then, $\langle E(U)e_n,e_m \rangle=\langle E(U)p_n(J)e_0,p_m(J)e_0\rangle=\langle p_m(J)E(U)p_n(J)e_0,e_0\rangle$. By Theorem \ref{orthop-1}, $J$ commutes with $E(U)$ and one can write 
\begin{equation}
\langle E(U)e_n,e_m \rangle=\langle p_m(J)p_n(J)E(U)e_0,e_0\rangle=\int_{U}p_n(t)p_m(t)d\mu(t).
\end{equation}
Hence the result.
\end{proof}

In physics oriented words, the above proposition states simply that the ''matrix elements'' of the spectral measure are given by the integrals of products of polynomials with respect to some positive compactly supported Borel measure. \\
From the above Proposition, it should be noted at this level that the spectral measure is entirely determined by $\mu$. In fact, the bridge between families of orthogonal polynomials and bounded Jacobi operators is provided by a corollary of the spectral theorem, known as the Favard theorem whose version for bounded operators will be given in a while. The following standard proposition set up the framework for the relevant orthogonal polynomials and associated 3-term recurrence relation.
\begin{proposition}\label{recur-poly}
Let $\{P_n\}_{n\in\mathbb{N}}$ be a family of polynomials of $\mathbb{R}[X]$, orthonormal in $L^2(\mu)$ with respect to a compactly supported measure $\mu$. There exists bounded sequences $\{a_n\}_{n\in\mathbb{N}}\subset\mathbb{R}^+$, $\{b_n\}_{n\in\mathbb{N}}\subset\mathbb{R}$ satisfying
\begin{eqnarray}
tP_n(t)&=&a_nP_{n+1}(t)+b_nP_n(t)+a_{n-1}P_{n-1}(t),\ k\ge1\nonumber\\
tP_0(t)&=&a_0P_1(t)+b_0P_0(t)\label{recurencepolygen}.
\end{eqnarray}
\end{proposition}
\begin{proof}
The fact that $\{P_n\}_{n\in\mathbb{N}}$ is an orthonormal family yields immediately $tP_n(t)=\sum_{p=0}^{n+1}\alpha_pP_p(t)$ with the only non-vanishing $\alpha_p$'s given by $\alpha_{n\pm1}=\int_\mathbb{R}tP_{n\pm1}(t)P_{n}(t)d\mu(t)$, $\alpha_n=\int_{\mathbb{R}}tP_n^2(t)d\mu(t)$. From this follows the recurrence \eqref{recurencepolygen} with $a_n=\alpha_{n+1}$, $b_n=\alpha_n$, $a_{n-1}=\alpha_{n-1}$ while the positivity of $a_n$ stems from the leading coefficient. \\
Next, one has 
\begin{equation}
|b_n|\le\int_\mathbb{R}|tP_n^2(t)|d\mu(t)\left(\sup_{t\in supp(\mu)}|t|\right)\times||P_n||^2_{L^2}\label{bnd1}. 
\end{equation}
On the other hand 
\begin{equation}
|a_n|\left(\sup_{t\in supp(\mu)}|t|\right)\int_\mathbb{R}|P_{n+1}(t)||P_n(t)|d\mu(t)\left(\sup_{t\in supp(\mu)}|t|\right)\times||P_{n+1}||_{L^2} ||P_n||_{L^2},\label{bnd2}
\end{equation}
where the last inequality stems from the use of H\"older inequality. From \eqref{bnd1} and \eqref{bnd2}, one concludes that $\{a_n\}_{n\in\mathbb{N}}$ and $\{b_n\}_{n\in\mathbb{N}}$ are bounded.
\end{proof}
\begin{remark} \ \\[-16pt]
\begin{itemize}
\item i) Notice that the boundedness of the sequences $\{a_n\}$ and $\{b_n\}$ explicitely uses the fact that the support of $\mu$, $supp(\mu)$ is compact. This corresponds to some specific families of orthogonal polynomials. This will have to be reconsidered in the case of unbounded Jacobi operators, therefore corresponding to other families of orthogonal polynomials.
\item ii) Notice also the similarity between \eqref{recurencepolygen} and \eqref{jacoboperator} which signals a link between those Jacobi operators and families of orthogonal polynomials determined by 2 sequences of real numbers as given above. This observation will be strengthened by the Favard theorem for bounded Jacobi operators, given below.
\end{itemize}
\end{remark}

\begin{theorem} [Favard Theorem]\label{Favard}
To each bounded Jacobi operator $J$ is associated a unique compactly supported measure $\mu$ satisfying:
\begin{itemize}
\vspace*{-4pt}
\setlength{\itemsep}{-1pt}
\item i) Orthonormal set:  For any polynomials $P\in\mathbb{R}[X]$, the map ${\cal{U}}:\ell^2(\mathbb{N})\to L^2(\mu)$ defined by ${\cal{U}}:P(J)e_0\mapsto P$ extends to a unitary operator and $\{P_n:={\cal{U}}e_n\}_{n\in\mathbb{N}}$ is an orthonormal family with respect to $\mu$.
\item ii) Intertwiner: Let $L:L^2(\mu)\to L^2(\mu)$ be the left multiplication operator by $t\in\mathbb{R}$. One has ${\cal{U}}J=L\cal{U}$.
\end{itemize}
\end{theorem}
\begin{proof}
First, one has for any $P_1, P_2\in\mathbb{R}[X]$ 
\begin{equation}
\langle P_1(J)e_0, P_2(J)e_0\rangle=\langle P_2^\dag(J)P_1(J)e_0, e_0\rangle=\int_\mathbb{R}P_1(t)P_2(t)d\mu(t),
\end{equation}
where the last equality stems from Proposition \ref{matrixelem-spectral-meas} and we use $E(\mathbb{R})=\bbone$. Now 
\begin{equation}
\int_\mathbb{R}P_1(t)P_2(t)d\mu(t)=\langle P_1P_2\rangle_{L^2}= \langle {\cal{U}}P_1e_0,{\cal{U}}P_2e_0\rangle_{L^2},
\end{equation}
where the definition of ${\cal{U}}$ has been used to obtain the last equality. Hence ${\cal{U}}^\dag{\cal{U}}={\cal{U}}{\cal{U}}^\dag=\bbone$. Next, from the beginning of the proof of Proposition \ref{matrixelem-spectral-meas}, $e_n=P_n(J)e_0$ implies that the subspace $\{J^ne_0,\ n\in\mathbb{N}\}$ of $\ell^2(\mathbb{N})$ on which ${\cal{U}}$ acts is dense in $\ell^2(\mathbb{N})$ while the image of ${\cal{U}}$ is dense in $\mathbb{R}[X]$ because $supp(\mu)$ is compact. Hence ${\cal{U}}$ extends to a unitary operator ${\cal{U}}:\ell^2(\mathbb{N})\to L^2(\mu)$. Finally, simply write 
\begin{equation}
\int_\mathbb{R}P_n(t)P_m(t)d\mu(t)=\langle {\cal{U}}e_n,{\cal{U}}e_m\rangle_{L^2}=\langle e_n,e_m\rangle=\delta_{nm},
\end{equation}
where we used the fact that ${\cal{U}}$ is unitary and i) is proven. \\
Now, from a standard computation using ${\cal{U}}Je_n=L{\cal{U}}e_n$ for any $n\in\mathbb{N}$ together with Propositions \ref{matrixelem-spectral-meas}, \ref{recur-poly} and Definition \ref{jacoboperator}, one obtains easily 
\begin{equation}
a_n=\langle Je_n,e_{n+1} \rangle=\int_\mathbb{R}tP_n(t)P_{n+1}(t)d\mu(t),\ b_n=\langle Je_n,e_n\rangle=\int_\mathbb{R}tP_n^2(t)d\mu(t). 
\end{equation}
Hence, ii) holds true. \\
Finally, let $s_n:=\int_\mathbb{R}x^nd\mu(x)$, $n\in\mathbb{N}$, be the moments of $\mu$ with the corresponding Stieltjes transform $w(z)=\int_\mathbb{R}\frac{1}{x-z}d\mu(x)$, for any $z\in\mathbb{C}\backslash\mathbb{R}$. Using the fact that $\mu$ is compactly supported, $supp(\mu)=[-\sigma,\sigma]$, one can write $w(z)=-\frac{1}{z}\sum_{n=0}^\infty\int_\mathbb{R}(\frac{x}{z})^n=-\sum_{n=0}^\infty\frac{s_n}{z^{n+1}}$ where the RHS is absolutely convergent whenever $|z|>\sigma$. Then, further using the formula 
\begin{equation}
\lim_{\delta\to0}\pi^{-1}\int_a^b\Im[ w(t+i\delta)]dt=\mu([a,b])+\frac{1}{2}(\mu(a)+\mu(b)), 
\end{equation}
one verifies that $\mu$ is determined uniquely by its moments. The theorem is proven.
\end{proof}
\begin{remark}
i) From Proposition \ref{measure-proba} and \ref{matrixelem-spectral-meas}, one infers $\langle E(U)e_0,e_0 \rangle=\int_UP_0(t)P_0(t)d\mu(t)$. By choosing $U=\mathbb{R}$, one obtains $\int_\mathbb{R}d\mu(t)=1$ provided $P_0=1$. Therefore, this last normalization condition, combined to the 3-term recurrence \eqref{recurencepolygen} of Proposition \ref{recur-poly} completely determines the set of polynomials $P_n$ and ensures that $\mu$ is a probability measure.\\
ii) The Favard theorem extends to the case of unbounded Jacobi operators, up to moderate modifications.\\
iii) We point out once more time that the above measure $\mu$ is entirely determined from the full sequence of its moments, $\{s_n:=\int_\mathbb{R}x^nd\mu(x) \}_{n\in\mathbb{N}}$. This deals with the so-called moment problem, another important block involved in the theory of orthogonal polynomials, which however is not directly used in the present paper. For more details, see e.g in \cite{akhiez:1965}, \cite{bsimon}.
\end{remark}
In the ensuing analysis, we will use at some places the Christoffel-Darboux relations. There are other algebraic formulas but we will not need them in this paper.
\begin{proposition}\label{chris-darb-rel}
Let $f^\prime$ denotes the derivative with respect to $x$. The following relation holds true: 
\begin{eqnarray}
(t-z)\sum_{k=0}^nP_k(t)P_k(z)&=&a_n(P_{n+1}(t)P_{n}(z)-P_{n}(t)P_{n+1}(z)),\label{christof-darb}\\
\sum_{k=0}^{N-1}(P_k(x))^2&=&P^\prime_N(x)P_{N-1}(x)-P^\prime_{N-1}(x)P_N(x)\label{christoff-darb2}.
\end{eqnarray}
\end{proposition}
\begin{proof}
This is routine computation. Multiply \eqref{recurencepolygen} by $P_{n}(z)$, then \eqref{recurencepolygen} in the $z$ variable by $P_n(t)$ and take the sum.
\end{proof}
Let us use a more physics oriented language for a while: In the following, the ''diagonalization'' of the self-adjoint (positive) kinetic operator of the considered NCFT will give rise to a 3-term recurrence equation satisfied by the ''matrix elements'' of a unitary (orthogonal) operator. This will therefore be related to a given family of orthogonal polynomials. \par  
Let us make a direct connection with NCFT for which the computation of the propagator may be a terrible task when carried out by brute force. But this computation can be actually simplified by using the following simple technical lemma which is a direct consequence of the Favard theorem together with the continuous functional calculus. It gives rise easily to explicit expression for the ''matrix elements'' of the propagator, once the family of orthogonal polynomials has been identified. Notice that it can be easily generalized to other situations (e.g $\ell^n(\mathbb{N})$ and/or unbounded self-adjoint Jacobi operators). Assume that the kinetic operator for a NCFT $J$ is a (bounded) Jacobi operator as in Definition \ref{defin-jacobi} with $J^{-1}$ such that $\{0\}\notin spec(J)$. According to the material we have presented above in this section, the following relations hold true for any $m,n\in\mathbb{N}$:
\begin{equation}
\langle e_m,J e_n \rangle=\int_{supp(\mu)}tP_m(t)P_n(t)d\mu(t),\label{formal-kinetic-operatoren}
\end{equation}
\begin{equation}
 \langle e_m,J^{-1} e_n\rangle=\int_{supp(\mu)}\frac{1}{t}P_m(t)P_n(t)d\mu(t)\label{formal-propagatoren}.
\end{equation}
where of course $\{P_n\}_{n\in\mathbb{N}}$ is the unique family of orthogonal polynomials associated to $J$. These relations are direct consequences of Definition \ref{defin-jacobi}, Proposition \ref{recur-poly} and e.g Proposition \ref{matrixelem-spectral-meas}. \par

To close this subsection, it is usefull to point out the relationship between the spectrum of a truncated $J$, i.e restricted to a $N\times N$ matrix, and the zeros of the related orthogonal polynomials. This will be helpfull in the process of characterization of the total spectrum of the kinetic operators, i.e point (discrete), continuous and residual spectra.
\begin{lemma}\label{lemma-spectral}
For a given $N\in\mathbb{N}$ let $J^N$ denotes the matrix of $\mathbb{M}_N(\mathbb{R})$ obtained from $J$ \eqref{jacoboperator} by restricting the indices $m,n,...$ to $0\le m,n,...\le N-1$. Then, the eigenvalues of $J^N$ are exactly given by the zeroes of the orthogonal polynomial $P_N$ in the orthogonal family defined by \eqref{recurencepolygen} related to $J$. The zeros of the orthogonal polynomials are simple and real.
\end{lemma}\label{lemma-diag-1}
\begin{proof}
For $m\le N$, the 3-term recurrence \eqref{recurencepolygen} can be easily cast into the matrix form 
\begin{equation}
t{\cal{P}}_N(t)=J^N{\cal{P}}_N(t)+a_NP_N(t)e_N,\ {\cal{P}}_N:=(P_0(t),P_1(t),...,P_{N-1}(t))\label{zero-eigen-jacob}.
\end{equation}
Whenever $t=t_0$, a zero of $P_N(t)$, \eqref{zero-eigen-jacob} becomes $t_0{\cal{P}}_N(t_0)=J^N{\cal{P}}_N(t_0)$. Now, $J^N$ is symmetric tridiagonal real matrix so it has real, simple eigenvalues. The lemma is proven.
\end{proof}
Many (but not all) properties presented here for the case of bounded Jacobi operators will still apply (up to minor adaptations) to the case of unbounded operators. Note by the way that the zeros of the orthogonal polynomials for bounded Jacobi all belong to $[-||J||,||J||]$ as a corollary of Lemma \ref{lemma-diag-1}.\\
The case of NCFT with unbounded Jacobi operators is slightly more complicated and will be presented elsewhere \cite{unboud-jac}. One difference comes from the ''large momentum'' behavior of the inverse of the kinetic operator, i.e the propagator, which decays in the latter case while it may remain non decaying in the bounded case, implying the occurrence of non zero correlation at large separation of indices in the language of matrix bases.\\
However (at least in the classes of noncommutative spaces we considered) the differences have a milder effect on the associated spectral triples (i.e basically the spectral triple with the square root of the kinetic operator as Dirac operator from which the NCFT action would derive as a spectral action). This comes mainly from the fact that the relevant non unital algebras are represented on a smooth (e.g Hilbert-Schmidt) subalgebra of some algebra of bounded operators. The latter representation enters explicitely the definition of some localized axioms for the non unital triple: functionals of the Dirac operator are (left or right) multiplied by smooth (e.g Hilbert-Schmidt) operator. This facilitates the obtention of compactness or summability conditions. In this respect, the bounded Jacobi operator case is not ''too far away'' from the unbounded case and is representative of common features. In this paper, the kinetic operator of the NCFT we consider is bounded Jacobi.\par


\section{\texorpdfstring{Gauge theories on $\mathbb{R}^2_\theta$ as matrix models}{Moyal 2-d}}\label{subsection22}


\subsection{Noncommutative gauge theories: A review on a nutshell}\label{section-review}

In this subsection, we will focus on some of the open problems so far left unsolved within noncommutative gauge theories built on what could be called informally  ''totally noncommutative geometries''. We will not consider the recent fascinating developments in gauge models on ''almost commutative geometries, in particular gauge models of Connes-Chamseddine-types that reproduce the main quantitative features of the Standard Model. See e.g 2nd of ref. \cite{Connes1}, \cite{ccs1}, \cite{ccs2}and references therein. \\
The NCFT studied in this paper is rigidely linked to a gauge theory on the noncommutative Moyal plane. It is therefore instructive to summarize the present situation for noncommutative gauge theories.\par
At the classical level, the construction of gauge invariant actions is not so difficult, once a differential calculus has been set up, together with a proper notion of noncommutative connection. This basically amounts to built a gauge invariant polynomial depending on some curvatures, usually a trace of products of curvatures. The use of spectral triples provides a natural way to construct noncommutative differential calculi as well as natural gauge invariant actions, once the spectral action principle is accepted. Another way is to use directly specific versions of the differential calculus based on derivations of the associative algebra related to the noncommutative space. This more algebraic way is often used in the mathematical physics literature as a flexible and fast tool to obtain expressions for noncommutative actions of NCFT. It extends easily to the construction of gauge invariant actions defined on ''noncommutative spaces'' and it appears that most of the developments in noncommutative gauge 
theories fit well within definite versions of derivation based differential calculi with the notion of noncommutative connection being a natural extension of the Koszul connection. This noncommutative differential calculus has been introduced a long ago in \cite{mdv88} (see also \cite{mdv99} and references therein). For further more recent developments and extensions see \cite{WAL1}, \cite{WAL2}, \cite{WAL3}.\par

Perturbative renormalisation of noncommutative gauge theories involves so far unsolved problems. Renormalisation of NCFT is in general more difficult than for commutative (local) field theories. This is often due to the UV/IR mixing. It occurs in the simplest noncommutative $\varphi^4$ model on the 4-d Moyal space, as noticed in \cite{Minwalla}, and destroys renormalisability. The phenomenon persists in noncommutative gauge models and is one of the main unsolved problems within Moyal NCFT. A 1st solution to this problem leading to renormalisable NCFT has been proposed in 2003 \cite{gw1}, \cite{gw2}. It is obtained by supplementing the scalar action with a harmonic oscillator term, leading to the now popular Grosse-Wulkenhaar model, whose 4-d version is likely non-perturbatively solvable, see \cite{harald-raimar}. The Grosse-Wulkenhaar model together with its gauge counterpart, to which we will turn on below, is actually related to an interesting NCG, called finite volume spectral triple and analyzed in \cite{
finite-vol}. It turns out the related noncommutative metric geometry is homothetic in a very precise sense to the usual noncommutative Moyal metric geometry, as shown in \cite{homot-moyal}.\par

At the present time, no (all orders) renormalisable gauge theory on 4-d Moyal space has been constructed, despite various encouraging progresses in that direction.
Inspired by the above Grosse-Wulkenhaar solution, a gauge model obtained either by heatkernel methods or by (equivalent) effective action computation, starting from of the Grosse-Wulkenhaar model coupled to a gauge field, gave rise to what is sometimes called the ''induced gauge theory'' on Moyal space \cite{GWW}, \cite{GW07}. The corresponding action is (the notations are by now standard):
\begin{align}
S=\int d^4x \Big(\frac{1}{4}F_{\mu\nu}\star F_{\mu\nu}
+\frac{\Omega^2}{4}\{{\cal{A}}_\mu,{\cal{A}}_\nu\}^2_\star
+{\kappa}{\cal{A}}_\mu\star{\cal{A}}_\mu\Big) \label{inducedgaugematrix}
\end{align}
in which  the 2nd and 3rd terms may be viewed as "gauge counterparts'' of the harmonic term of the Grosse-Wulkenhaar model, while the first term that looks like a Yang-Mills action has bad UV/IR mixing already showing up as a hard IR singularity in the vacuum polarization tensor. Here, ${\cal{A}}_\mu$ is the covariant coordinates, a natural gauge covariant tensor form stemming from the existence of a canonical gauge invariant connection in the present NC framework. We refer to \cite{WAL1}, \cite{WAL2} for a complete description of the underlying noncommutative differential calculus and relevant notion of noncommutative connection. Here, it is sufficient to note that the relevant gauge transformation are 
\begin{equation}
{\cal{A}}_\mu\rightarrow{\cal{A}}_\mu^g=g^\dag\star{\cal{A}}_\mu\star g,\ \forall g\in{\cal{U}}(\mathbb{R}^4_\theta)
\end{equation}
with similar covariant transformation for the curvature $F_{\mu\nu}$. Here $\mathbb{R}^4_\theta$ is identified in this algebraic framework to the multiplier algebra of the Schwartz functions algebra with associative Moyal product and ${\cal{U}}(\mathbb{R}^2_\theta)$ denotes the corresponding unitary elements. The structure of the (symmetric) vacua of \eqref{inducedgaugematrix} has been analyzed in detail and classified in \cite{GWW2}. It was realized from this work that the complicated vacuum structure thus exhibited forbids a direct perturbative treatment of the gauge action \eqref{inducedgaugematrix}. \par

Alternative way of introducing a Grosse-Wulkenhaar term within a gauge invariant theory has been also proposed. These are based on the implementation of a IR damping mechanism \cite{blaschk1}. In particular in the 1st ref. of \cite{blaschk1}, it has been claimed that renormalisability of a noncommutative gauge-invariant model could be restored by building an IR damping into the propagator for the gauge potential propagator together with nonlocal counterterms related to the (quadratic and linear) IR singularities triggering the UV/IR mixing, akin to the ''soft-breaking'' terms of the Gribov-Zwanziger action. In \cite{bgw-13}, Slavnov-Taylor identities have been used to track the IR singularities in gauge models of the above type, showing basically that the needed independent Gribov-type parameters reduce in fact to one. This damping approach is appealing. However, interpreting the action within the framework of a definite noncommutative differential geometry is unclear at the present time, unlike the case of 
the induced gauge action.\par 

We mention finally that modifications of the noncommutative differential calculus underlying the action \eqref{inducedgaugematrix} have been investigated in order to tackle the UV/IR mixing, resulting in modifications of \eqref{inducedgaugematrix} \cite{WAL1,WAL2}. Although interesting relationships have been shown between these modifications and the Grosse-Wulkenhaar model, no improvement of the UV/IR mixing behavior has been obtained. Notice that the inclusion of fermions is expected to improve the situation. 

Another appealing (in some sense dual) approach is provided by the matrix model formulation of noncommutative gauge theory has also been investigated, initiated a long ago in \cite{matrix1} in the context of type IIB (stringy) matrix models. This basically amounts to re-interpret the noncommutative gauge theories as matrix models taking advantage of the relationship between the gauge potential and the covariant coordinate mentionned above, exhibiting in some cases a relationship with the action \eqref{inducedgaugematrix}. For recent exhaustive reviews, see \cite{matrix2} (see also \cite{matrix3}). Note that the matrix model approach may in some cases allow one to go beyond the perturbative approach \cite{matrix4}. One interesting outcome is that it may provide a interpretation for the UV/IR mixing in some noncommutative gauge theories in terms of an induced gravity action. See e.g \cite{matrix5}. \par

Motivated by this alternative approach, the induced gauge theory action \eqref{inducedgaugematrix} on the Moyal plane, expanded around a particular symmetric vacuum, has been formulated as a matrix model and its 1-loop behavior has been studied \cite{MVW13}. This analysis provides the first investigation of the quantum properties of \eqref{inducedgaugematrix}, showing by the way the usefulness of the matrix model formulation to overcome the long standing difficulty linked to the vacuum structure. One of the result is the appearance of a non vanishing 1-point (tadpole) function indicating a quantum vacuum instability. Regardless its gauge theory origin, this model is essentially a NCFT with a particular bounded Jacobi operator as kinetic operator and polynomial interactions after suitable gauge fixing and is actually representative of most of the main features of any other NCFT with other bounded Jacobi kinetic operator. Note also that some (but not all!) features of this bounded case will be somewhat similar 
to the unbounded case. This is why we have chosen to investigate deeply the above quantum matrix gauge theory on the Moyal plane. This will be done in the sequel. Notice by the way that this gauge model is related to the 6-vertex model (to which it reduces whenever the vacuum is the trivial one \cite{MVW13}).\par 

More recently, we have investigated gauge theories on $\mathbb{R}^3_\lambda$ \cite{gvw13}. While non-vanishing 1-point function also occurs at the one-loop order, deserving further investigation, it appears that the UV behavior as well as the IR one are milder than what one would obtain for a NCFT on (4-d) Moyal space. This behavior was also observed in the families of scalar NCFT on $\mathbb{R}^3_\lambda$ considered in \cite{vit-wal-12}, among which classes of finite NCFT have been identified, leading to the conclusion that no perturbatively dangerous UV/IR mixing should occur within these scalar NCFT. One could object that it just comes from dimensional effect as it happens for commutative field theories and/or that it reflects a particular choice of the differential calculus. This is not correct. In fact, a closer inspection of inspection of $\mathbb{R}^3_\lambda$ shows that it has a natural decomposition as $\mathbb{R}^3_\lambda\simeq\oplus_{j\in\frac{\mathbb{N}}{2}}\mathbb{M}_{2j+1}(\mathbb{C})$, 
reflecting the fact that $\mathbb{R}^3_\lambda$  is a closure of $U(su(2))$, the universal envelopping algebra of $su(2)$. This yields one-loop amplitudes in which $j$, the radius of the fuzzy sphere $\mathbb{S}_j=\mathbb{M}_{2j+1}(\mathbb{C})$ in the physicists language, plays the role of a natural UV cut-off. Two important comments are in order. First, the models considered in \cite{vit-wal-12} and \cite{gvw13} are {\it{not}} models on fuzzy sphere. Only suitable truncations (reduction in the physicists language) of these models to a single fuzzy sphere are! Next, the noncommutative differential calculus used in \cite{vit-wal-12} and \cite{gvw13}, although perfectly admissible, is not the most natural one. In fact, there is a natural action of $\mathbb{C}(SU(2))$ onto $\mathbb{R}^3_\lambda$, stemming from the duality pairing between $U(su(2))$ and $\mathbb{C}(SU(2))$. This yields a bicovariant differential calculus on $\mathbb{R}^3_\lambda$ with a Laplacian whose commutative limit is the usual Laplacian on 
$\mathbb{R}^3$. These observations identify ways of future investigations. It should be interesting to see in what extend the above mild UV and IR behavior is affected by the change of differential calculus and to see how this behavior survives at higher orders in perturbation. It would then favor an interesting more general class of noncommutative spaces leading to nicely behaving NCFT, to which $\mathbb{R}^3_\lambda$ belongs as an element. Notice that the gauge theory on $\mathbb{R}^3_\lambda$ studied in \cite{gvw13} when reduced to a single fuzzy sphere in $\mathbb{R}^3_\lambda$, can be identified with the Alekseeev-Recknagel-Schomerus action \cite{ARS} describing some low energy action for brane dynamics on $\mathbb{S}^3$.


\subsection{\texorpdfstring{Quantum matrix gauge theory on $\mathbb{R}^2_\theta$}{Moyal 2-d}}

We begin this subsection by synthesizing the properties of the Moyal plane that are actually needed in the subsequent analysis. There is a huge literature on Moyal spaces. For a complete presentation, see in e.g \cite{GBVF} and references therein.\\
Informally, the Moyal plane can be viewed as a strict quantization of $\mathbb{R}^2$ to the Poisson bracket on $C_0(\mathbb{R}^2)$. Behind the common vocable ''Moyal algebra'', there are various algebras depending on the required properties of regularity and/or differentiability. To deal with NCFT, it is especially convenient to view the relevant Moyal algebra as already represented on a suitable $C^*$-algebra of bounded operators on some Hilbert space. This is the viewpoint we adopt here, which permits one to avoid cumbersome manipulations of integrals of star-products arising when working with (Frechet) algebras of functions. The simplification is especially efficient in the case of the noncommutative space $\mathbb{R}^3_\lambda$ for which the measure in the relevant integral is not the usual Lebesgues measure on $\mathbb{R}^3_\lambda$, but has a complicated expression.\par

Let ${\cal{S}}$ be the space of $\mathbb{C}$-valued Schwartz functions on $\mathbb{R}^2$ and $\{f_{mn}(x) \}_{m,n\in\mathbb{N}}\subset{\cal{S}}\subset L^2(\mathbb{R}^2)$ be the family of Wigner transition eigenfunctions of the 1-d harmonic oscillator, with ($\theta>0$)
\begin{equation}
f_{mn}\star f_{kl}=\delta_{nk}f_{ml},\ f_{mn}^\dag=f_{nm},\ \langle f_{mn},f_{kl} \rangle_{L^2}=2\pi\theta\delta_{mk}\delta_{nl}\label{matrix-base-basic},
\end{equation}
where $\star$ is the Moyal product whose explicit expression will not be needed here. Let $(e_n)_{n\in\mathbb{N}}$ be the canonical basis of $\ell^2(\mathbb{N})$, see subsection \ref{subsection21}. First, define the following Frechet $*$-algebras
\begin{eqnarray}
{\mathbb{A}}:&=&\left\{\phi\in({\cal{S}},\star),\ \rho_{\alpha,\beta}(\phi)=\sup_{x\in\mathbb{R}^2}|x_1^{\alpha_1}x_2^{\alpha_2}\partial_1^{\beta_1}\partial_2^{\beta_2}\phi|<\infty,\ \forall \alpha_i,\beta_i\in\mathbb{R}^+ \right\}\label{alg1}\\
{\cal{M}}_\theta:&=&\left\{(\Phi_{m,n})_{m,n\in\mathbb{N}},\ \rho_n(\Phi)=\sum_{p,k}\theta(p+\frac{1}{2})^n(k+\frac{1}{2})^n|\Phi_{pk}|^2<\infty,\ \forall n\in\mathbb{N}\right\}\label{alg2},
\end{eqnarray}
i.e \eqref{alg2} is the Frechet subalgebra of $\ell^2(\mathbb{N}^2)$ involving the rapid decay matrices (equipped with matrix product).\par
It is known that ${\mathbb{A}}$ is isomorphic to ${\cal{M}}_\theta$ as Frechet $*$-algebras, with isomorphism defined by $\Phi_{mn}\mapsto \sum_{m,n}\Phi_{mn}f_{mn}\in{\cal{S}}$ and inverse $\phi\in{\cal{S}}\mapsto\frac{1}{2\pi\theta}\langle \Phi,f_{mn}\rangle_{L^2}$. Next, ${\cal{M}}_\theta$ is represented on a subalgebra of the $C^*$-algebra  ${\cal{B}}(\ell^2(\mathbb{N}))$ by using the natural representation of matrices as operators on $\ell^2(\mathbb{N})$, denoted by $\eta$. This is simply the product of a matrix by a column vector. Indeed, for any $\Phi\in{\cal{M}}_\theta$, $\Phi=\sum_{m,n}\Phi_{mn}e_m\otimes e_n$, define
\begin{equation}
\eta:\ell^2(\mathbb{N})\otimes\ell^2(\mathbb{N})\to{\cal{B}}(\ell^2(\mathbb{N})),\ \eta(e_m\otimes e_n)=e_m\otimes e^*_n,\ \forall m,n\in\mathbb{N}\nonumber
\end{equation}
with dual $e^*_n$ such that $e^*_n(e_p)=\delta_{np}$. In order to make contact with the physicists notations, we will set $\eta(e_m\otimes e_n)=e_m\otimes e^*_n:=|m\rangle\langle n|=f_{mn}$ in a while. \par
The definition of ${\cal{M}}_\theta$ \eqref{alg2} implies that the operator $\eta(\Phi)$ is Hilbert-Schmidt (set $n=0$ in the seminorm $\rho_n(\Phi)$). Hence $\eta(\Phi)$ is compact. Thus $\eta({\cal{M}}_\theta)$ inherits a $C^*$-norm from its Hilbert-Schmidt action on $\ell^2(\mathbb{N})$ (and so does ${\cal{M}}_\theta$). One concludes that the closure of $\eta({\cal{M}}_\theta)$ is such that $\overline{\eta({\cal{M}}_\theta)}\subseteq{\cal{K}}(\ell^2(\mathbb{N}))$ where ${\cal{K}}(\ell^2(\mathbb{N}))$ is the $C^*$-subalgebra of compact operators on $\ell^2(\mathbb{N})$. But since $\eta$ is faithfull, one has the isomorphism $\eta({\cal{M}}_\theta)\simeq{\cal{M}}_\theta$. Hence $\overline{{\cal{M}}_\theta}\subset{\cal{K}}(\ell^2(\mathbb{N}))$. On the other hand, $\overline{{\cal{M}}_\theta}\supseteq{\cal{K}}(\ell^2(\mathbb{N}))$ since ${\cal{M}}_\theta\supseteq\bigcup_{n=1}^\infty{\mathbb{M}}_n(\mathbb{C})$ holds true. Thus $\overline{{\cal{M}}_\theta}={\cal{K}}(\ell^2(\mathbb{N}))$. Finally, the map $U:L^2(
\mathbb{R}^2)\to\ell^2(\mathbb{N})\otimes\ell^2(\mathbb{N})$ defined by $U:f_{mn}\mapsto {\sqrt{2\pi\theta}}e_m\otimes e_n$ for any $m,n\in\mathbb{N}$ is an isometry. Hence $\overline{{\mathbb{A}}}\simeq{\cal{K}}(\ell^2(\mathbb{N}))$. Thus, we are dealing with compact operators at the end of the day.\par

In the following, it is understood that the algebra ${\mathbb{A}}$ is already represented on ${\cal{B}}(\ell^2(\mathbb{N}))$ as we have explained above. The resulting algebra will be sometimes denoted by $\mathbb{R}^2_\theta$ for short. According to the above discussion, we drop the $\eta$ symbol everywhere and denote any element in $\mathbb{R}^2_\theta$ (that corresponds actually to a field in the NCFT framework) as $\Phi=\sum_{m,n\in\mathbb{N}}\Phi_{mn}|m\rangle\langle n|$, where the $|m\rangle\langle n|$'s correspond to the representation of the $f_{mn}$'s on ${\cal{B}}(\ell^2(\mathbb{N}))$. To simplify the notations, we will from now on set $|m\rangle\langle n|=f_{mn}$ (i.e use the same notation for, says operators and corresponding symbols). Of course, the Moyal $\star$-product becomes formally a ''matrix product'' (since $\eta$ is a $*$-algebra morphism, $\eta(MN)=\eta(M)\eta(N)$) and the above relations for the $f_{mn}$'s still hold true among the represented $f_{mn}$ while 
\begin{equation}
\int_{\mathbb{R}^2}\to\tr. \nonumber
\end{equation}
Note that the Moyal product can be extended to larger algebras, by exploiting families of seminorms, see e.g \cite{Gracia-Bondia:1987kw} for an extensive discussion.\par

In the following, we use the covariant coordinate ${\cal{A}}_\mu$  defined above as the fundamental field variable in the action, i.e $S=S[{\cal{A}}_\mu]$ which amounts to interpret the induced gauge theory \eqref{inducedgaugematrix} as a matrix model. This latter is invariant under the gauge transformations ${\cal{A}}_\mu\rightarrow{\cal{A}}_\mu^g=g^\dag\star{\cal{A}}_\mu\star g,\ \forall g\in{\cal{U}}(\mathbb{R}^2_\theta)$. Notice by the way the formal similarity between the 1st term of the action \eqref{inducedgaugematrix} at $\Omega=0$ and the bosonic part of the IKKT matrix model (see e.g \cite{blasch-stein-10} and references therein). The suitable differential calculus underlying the present functional action can be viewed as the simplest differential calculus \`a la Koszul generated by the 2 natural derivations of $\mathbb{R}^2_\theta$. This has been discussed and analyzed at length in \cite{WAL1,WAL2} to which we refer for details. The only tools we need for the moment are the gauge transformations 
of the covariant coordinates given above and general features of the natural basis $(f_{mn}),\ m,n\in\mathbb{N}$ of $\mathbb{R}^2_\theta$, summarized by \eqref{matrix-base-basic}. The 2-dimensional version of \eqref{inducedgaugematrix} is
\begin{equation}
S_\Omega[\mathcal{A}]=\tr\big((1+\Omega^2){\cal{A}}{\cal{A}^\dag}{\cal{A}}{\cal{A}}^\dag+(3\Omega^2-1){\cal{A}}{\cal{A}}{\cal{A}}^\dag{\cal{A}}^\dag+
2\kappa{\cal{A}}
{\cal{A}}^\dag\big)\label{classaction},
\end{equation}
where we define ${\cal{A}}={{{\cal{A}}_1+i{\cal{A}}_2}\over{\sqrt{2}}},\quad {{\cal{A}}}^\dag={{{\cal{A}}_1-i{\cal{A}}_2}\over{\sqrt{2}}}$.
Notice that the functional action \eqref{classaction} has some similarities with the 6-vertex model\footnote{This similarity has been pointed out to us by H. Steinacker. However, the actual analysis and the actual ''identity'' of the model depends on the choice of a particular vacuum around which the classical theory is expanded, as illustrated below.}. The corresponding equation of motion is 
\begin{equation}
(3\Omega^2-1)({\cal{A}}^\dag{\cal{A}}{\cal{A}}+{\cal{A}}{\cal{A}}{\cal{A}}^\dag)+2(1+\Omega^2){\cal{A}}{\cal{A}}^\dag{\cal{A}}+2\kappa{\cal{A}}=0\label{eqmotion}.
\end{equation}
We now have to choose a particular vacuum (the background), expand the action around it, fix the background symmetry of the expanded action. Focusing on vacuum configurations with the global symmetries of the classical action as in \cite{GWW2} amounts to look for solutions $Z$ of the equation of motion with the following expansion
\begin{equation}
Z=\sum_{m,n\in\mathbb{N}}Z_{mn}f_{mn},\ Z_{mn}=-ia_{m}\delta_{m+1,n}, \forall m,n\in\mathbb{N}\label{vacuummatrix}
\end{equation}
where  the sequence of complex numbers $\{a_m, m\in\mathbb{N}\}$ satisfies
\begin{equation}
a_m\bigg((3\Omega^2-1)(|a_{m+1}|^2+|a_{m-1}|^2)+2(1+\Omega^2)|a_m|^2+2\kappa  \bigg)=0,\label{recursive}
\end{equation}
obtained by plugging $Z$ \eqref{vacuummatrix} into \eqref{eqmotion}. The trivial vacuum solution is $a_m=0$. The corresponding model would then be (a version of) the 6-vertex. We turn now on the non trivial symmetric vacua classified in \cite{GWW2}. They are given, up to an unessential phase ($\xi_m$ in the notations of \cite{GWW2}), by $a_m=\sqrt{u_m}$ with
\begin{eqnarray}
0<\Omega^2<\frac{1}{3}&,&\ u^1_m=\alpha(r^m-r^{-m})-\frac{\kappa}{4\Omega^2}(1-r^{-m}),\ \alpha\ge0,\ r>1\label{vac11}\\
\Omega^2={{1}\over{3}}&,&\ \kappa<0,\ a_m={{1}\over{2}}{\sqrt{ -3\kappa}},\ \forall m\in\mathbb{N}\label{sol1}\\
\frac{1}{3}<\Omega^2<1&,&\ u^2_m=-\frac{\kappa}{4\Omega^2}(1-r^{-m}),\ \kappa\le0,\ r\le-1\label{vac12}\\
\Omega^2=1&,&\ u^3_m=-\frac{\kappa}{4}(1-(-1)^{-m})\label{vac13},\ \kappa\le0
\end{eqnarray}
where
\begin{equation}
r=\frac{1+\Omega^2+\sqrt{8\Omega^2(1-\Omega^2)}}{1-3\Omega^2}\label{err}.
\end{equation}
Within the present phase choice, $a_m\in\mathbb{R}$. Let us summarize the main steps of the derivation of the gauge fixed action performed in \cite{MVW13}. By setting formally ${\cal{A}}=Z+\phi$, ${\cal{A}}^\dag=Z^\dag+\phi^\dag$ into \eqref{classaction}, where $\phi$ can be interpreted as a fluctuation around $Z$, one easily obtains a functional action $S[\phi,\phi^\dag]$ invariant under a background transformation related to a nilpotent BRST-like operation, $\delta_Z$, with structure equations:
\begin{equation}
\delta_Z\phi=-i[Z+\phi,C],\ \delta_Z\phi^\dag=-i[Z^\dag+\phi^\dag,C],\ \delta_ZZ=0,\ \delta_ZC=iCC\label{brsbackgrd}.
\end{equation}
Here, $C$, ${\bar{C}}$ and $b$ are respectively the ghost, the antighost and the St\"uckelberg field with ghost number equal to $+1$, $-1$ and $0$. $\delta_Z$ acts as a graded derivation with grading equal to the sum of the degree of forms and ghost number (modulo 2) and $\delta_Z^2=0$. This background symmetry can be fixed by supplementing $S[\phi,\phi^\dag]$ with the gauge-fixing action
\begin{equation}
S_{GF}=\delta_Z\tr\ {\bar{C}}(\phi-\phi^\dag)=\tr\ \bigg( b(\phi-\phi^\dag)+i{\bar{C}}[Z-Z^\dag+\phi-\phi^\dag,C]\bigg)\label{gaugefixingphi},
\end{equation}
where $\delta_Z{\bar{C}}=b$, $\delta_Zb=0$. There are interesting underlying algebraic structures which have been investigated in \cite{rsw}. Upon integrating the $b$ field, the ghost fields decouple. Setting now $\phi=\sum_{m,n}\phi_{mn}f_{mn}$, the remaining (non-ghost) part of the gauge fixed action reduces to a functional of $\phi$ only given by
\begin{equation}
S[\phi]=\sum_{m,n,k,l\in\mathbb{N}} \phi_{mn}\phi_{kl}G_{mn;kl}+S_{int}\label{actmatrix}
\end{equation}
where the kinetic operator takes the complicated expression
\begin{eqnarray}
G_{mn;kl} &=&(1+5\Omega^2)\delta_{ml}\delta_{nk}(a_na_{n+1}+a_na_{n-1})\nonumber\\
&-&(3\Omega^2-1)(\delta_{ml}\delta_{n+1,k-1}a_na_{n+1}+\delta_{ml}\delta_{n-1,k+1}a_na_{n-1}
-2\delta_{m,l+1}\delta_{k+1,n}a_na_l)\nonumber \\
&-&(1+\Omega^2)(\delta_{k,n+1}\delta_{m,l+1}a_na_l+\delta_{n,k+1}\delta_{l,m+1}a_na_l)+2\kappa\delta_{ml}\delta_{nk}\label{kineticoperator1}
\end{eqnarray}
with the $a_m$'s given by any of the sequences defined in \eqref{vac11}-\eqref{err}. \par
The kinetic operator involves 2 different types of terms: i) terms proportional to $(3\Omega^2-1)$ which are non zero when $m+n=l+k\pm2$, ii) terms that are non zero when $m+n=l+k$ which is similar to the so called conservation law for indices of the Grosse-Wulkenhaar model. The cubic and quartic interaction terms are
\begin{equation}
S_{int}=8\Omega^2\sum_{m,p,q,r\in\mathbb{N}}i\phi_{pq}\phi_{qr}\phi_{mp}(a_r\delta_{m+1,r}-a_r\delta_{r+1,m})+
4\Omega^2\sum_{m,n,k,r\in\mathbb{N}}\phi_{mn}\phi_{nk}\phi_{kr}\phi_{rm}\label{interactmatrix}.
\end{equation}
At this point, one remark is in order. Eqn. \eqref{kineticoperator1} shows clearly that the kinetic operator, hence the dynamics coded by the model, depends essentially on the choosen vacuum.\par
From now on, we choose the vacuum given by \eqref{sol1}, namely: 
\begin{equation}
\Omega^2={{1}\over{3}},\ \kappa<0,\ a_m={{1}\over{2}}{\sqrt{ -3\kappa}},\ \forall m\in\mathbb{N}\nonumber.
\end{equation}
for which the analysis is slightly simpler. The kinetic operator \eqref{kineticoperator1} becomes
\begin{equation}
 G^{(1/3)}_{mn;kl}=(-\kappa)\big(2\delta_{ml}\delta_{nk} - \delta_{k,n+1}\delta_{m,l+1}-\delta_{n,k+1}\delta_{l,m+1}\big)\label{kinetic13},
\end{equation}
and fullfills
\begin{equation}
G^{(1/3)}_{mn;kl}\ne0\iff m+n=k+l\label{conservation}.
\end{equation}
Then, upon setting $n=\alpha-m$, $k=\alpha-l$, with $\alpha=m+n=k+l$ into \eqref{kinetic13}, the indice conservation law \eqref{conservation} leads to
\begin{equation}
G^{(1/3)}_{m,\alpha-m;\alpha-l,l}:=G^{\alpha}_{m,l}=\mu^2(2\delta_{ml}- \delta_{m,l+1}-\delta_{l,m+1}),\ \forall m,l\in\mathbb{N}\label{tridiagon}
\end{equation}
with $\mu^2:=-\kappa$, which is independent of $\alpha$. Some (spectral) properties of the kinetic operator defined by \eqref{tridiagon} can now be easily obtained by a mere application of the material developed in subsection \ref{subsection21}.
\begin{proposition}\label{kinetic-G}
Let $G_{nm}:=G^{\alpha}_{n,m}$ be the kinetic operator on $\ell^2(\mathbb{N})$ defined by \eqref{tridiagon}. The following properties hold true:
\begin{itemize}
\vspace*{-4pt}
\setlength{\itemsep}{-1pt}
\item i) The operator on $\ell^2(\mathbb{N})$ given by $(-G)_{nm}$ defines a bounded (thus self-adjoint) Jacobi operator on $\ell^2(\mathbb{N})$.
\item ii) The spectrum of the truncated kinetic operator $G^N_{nm}$, $0\le n,m\le N-1$ is given by 
\begin{equation}
spec(G^N)=\left\{\lambda^N_k:=2\mu^2\left(1-\cos({{(k+1)\pi}\over{N+1}})\right), \ k\in\{0,1,...,N-1\}, N\in\mathbb{N}^*\right\}.\label{spectrum1}
\end{equation}
\item iii) The spectrum of the kinetic operator $G$ is simple.
\end{itemize}
\end{proposition}
\begin{proof}
Set $J:=-\frac{1}{\mu^2}G$. Using \eqref{tridiagon}, one has $(Je)_n=e_{n+1}+e_{n-1}-2e_n$, $(Je)_0=e_1-2e_0$. Then, simply use Definition \ref{defin-jacobi} and Propositions \ref{extend-self} and \ref{jaco-bounded} with $a_n=1$, $b_n=-2$ for any $n\in\mathbb{N}$ and i) is proven.\\
To prove ii), one first determine the associated family of orthogonal polynomials. Then, Favard Theorem \ref{Favard} and Proposition \ref{recur-poly} guaranty the 
existence of a family of orthogonal polynomials with respect to a unique compactly supported probability measure with 3-term recurrence given by
\begin{equation}
P_{n+1}(t)+P_{n-1}(t)-2P_n(t)=tP_n(t),\ P_1(t)-2P_0(t)=tP_0(t),\ P_0(t)=1.
\end{equation}
Now set $2x=2+t$ and $P_n(2x-2)=U_n(x)$. Then, one obtains
\begin{equation}
U_{n+1}(x)+U_{n-1}(x)=2xU_n(x),\ U_1(x)=2xU_0(x),\ U_{-1}=0,\ U_0(x)=1\label{chebysh2}
\end{equation}
which, using the Askey classification (see \cite{kks}), defines a 3-term recurrence for the Chebyschev polynomials of 2nd kind \cite{kks}. From Lemma \ref{lemma-diag-1}, it follows that the eigenvalues of $J^N$ obtained by truncating the indices $n,m,...$ as $0\le n,m,...\le N-1$ are exactly given by the zeros of the Chebyshev polynomials $U_N(x)$. These latter are given \cite{kks} by $z_k^N=\cos({{(k+1)\pi}\over{N+1}})$, $k=0,1,...,N-1$. Therefore, the zeros of $P_N(2x-2)$ are given by $(2z^k_N-1)$, $k=0,1,...,N-1$. This together with $J:=-\frac{1}{\mu^2}G^N$ yields \eqref{spectrum1} and ii) is proven. \\
Finally, iii) is guaranteed by a theorem due to M. Stone (see in \cite{akhiez:1965}) stating that any self-adjoint operator with simple spectrum is generated by a Jacobi operator of the type considered here. This completes the proof.
\end{proof}
\begin{remark}
Let us first collect some technical (albeit important) points. We recall that Chebyshev polynomials of second kind are given by
\begin{equation}
U_n(t):=(n+1){\mbox{$_2$F$_1$}(-n,\ n+2;\ {{3}\over{2}};\ {{1-t}\over{2}})}={{P_n^{{{1}\over{2}},{{1}\over{2}}}(t) }\over{P_n^{{{1}\over{2}},{{1}\over{2}}}(1) }}, \ \forall n\in\mathbb{N}\label{defcheb},
\end{equation}
where $_2$F$_1$ denotes the hypergeometric function and $P_n^{\alpha,\beta}(x)$ is a particular family of Jacobi polynomials. Recall that for any $N\in\mathbb{N}$, $U_N(t)$ has $N$ different simple roots in $[-1,1]$. The orthogonality relation among the Chebyshev polynomials $U_n$ is
\begin{equation}
\int_{-1}^1d\mu(x)\ U_m(x)U_n(x)={{\pi}\over{2}}\delta_{mn},\ d\mu(x)=dx{\sqrt{1-x^2}}\label{orthocontinucheb2}
\end{equation}
For more details, see e.g \cite{kks}. \\
To close this remark, we mention that the Jacobi operator considered here is not completely continuous (although it is of course continuous). Recall that a completely continuous bounded operator maps any weakly convergent sequence into a strongly convergent one. As far as bounded Jacobi operators are concerned, complete continuity is equivalent to the condition that the sequences $a_n$ and $b_n$ go to zero as $n\to\infty$. This condition is not fullfilled here.
\end{remark}
\begin{proposition}\label{g-positive}
The kinetic operator $G$ on $\ell^2(\mathbb{N})$ \eqref{tridiagon} is positive.
\end{proposition}
\begin{proof}
For any $v\in\ell^2(\mathbb{N})$, $v=\sum_kv_ke_k\ne0$, we compute $\langle v,Gv\rangle$. One has:
\begin{equation}
\langle v,Gv\rangle=\mu^2\sum_{m\in\mathbb{N}}(v_m^\dag(2v_m-v_{m+1}-v_{m-1}))=\mu^2\left(|v_0|^2+\sum_{m\in\mathbb{N}}|(v_m-v_{m+1})|^2\right)>0.
\end{equation}
Hence, $G$ is positive.
\end{proof}
\begin{remark}\label{locate-spectrum} \ \\[-20pt]
\begin{itemize}
\setlength{\itemsep}{-1pt}
\item i) Since $G$ is positive, there exists a self-adjoint operator $D$ on $\ell^2(\mathbb{N})$ such that $G=D^2$. $D$ may therefore play the role of a Dirac operator, that is bounded since $G$ is bounded. We will examine below the spectral triple built from the standard Moyal triple in the ''matrix base'' formulation where the usual Dirac operator, says $\slashed\partial$ is replaced by the above operator $D$.
\item ii) Proposition \ref{g-positive}, together with the fact that $G\in{\cal{B}}(\ell^2(\mathbb{N}))$ imply that the spectrum of $G$, $spec(G)$, satisfies $spec(G)\subseteq[0,||G||]$. We now characterize completely the spectrum of $G$.
\end{itemize}
\end{remark}

We first characterize the point spectrum for the operator $G$, $spec_P(G)$ then use a general theorem valid for bounded self-adjoint operator to identify the full spectrum. This can be easily performed by simply remarking that the eigenvectors of $G^N$ that gave rise to \eqref{spectrum1} can be used to built (''finite range'') eigenvectors for the operator $G$ on $\ell^2(\mathbb{N})$, as a simple corrolary of Lemma \ref{lemma-diag-1}. Namely, 
\begin{lemma}\label{lemma-spectrum-full} \ \\[-20pt]
\begin{itemize}
\setlength{\itemsep}{-1pt}
\item i) The point spectrum of the kinetic operator $G$ on $\ell^2(\mathbb{N})$ defined by \eqref{tridiagon} is $spec_P(G)=\cup_{N\in\mathbb{N}^*}spec(G^N)$ where $spec(G^N)$ given by \eqref{spectrum1} with family of related orthonormal eigenvectors of $\ell^2(\mathbb{N})$ given by 
\begin{equation}
v_{N,m}=\sum_{p=0}^{N-1}\bigg((-1)^m(N+1){{ \sin[{{N(m+1)\pi }\over{N+1 }}  ]}\over {\sin^3[{{(m+1)\pi }\over{N+1 }} ] }} \bigg)^{-{{1}\over{2}}}  U_p\left({{2+\lambda_{N,m}}\over{2}}\right)e_p\label{eigenvektor},
\end{equation}
where $\{e_p\}_{p\in\mathbb{N}}$ is the canonical basis of $\ell^2(\mathbb{N})$ and corresponding eigenvalues
\begin{equation}
\lambda_{N,m}:=2\left(\cos\left({{(m+1)\pi}\over{N+1}}\right)-1\right), \ m\in\{0,1,...,N-1\}, N\in\mathbb{N}^*\}.
\end{equation}
with multiplicity 1.
\item ii) The spectrum of $G$ is 
\begin{equation}
spec(G)=\left\{\lambda^N_k,\ k\in\{0,1,2,...,N-1\}, N\in\mathbb{N}, N\ne0 \right\}\cup\{0\}\label{spectrum-full}.
\end{equation}
\end{itemize}
\end{lemma}
\begin{proof}
First, recall that the point spectrum is determined by the eigenvalues of $G$. We set 
\begin{equation}
G_{ml}=\mu^2\sum_{p\in\mathbb{N}}{\cal{R}}_{mp}\lambda_p{\cal{R}}^\dag_{pl},\ \sum_{p\in\mathbb{N}}{\cal{R}}_{mp}{\cal{R}}^\dag_{pl}=\sum_{p\in\mathbb{N}}{\cal{R}}^\dag_{mp}{\cal{R}}_{pl}=\delta_{ml},\ {\cal{R}}^\dag_{mn}={\cal{R}}_{nm}\label{unitarytrans}.
\end{equation}
From this one easily obtains ${\cal{R}}_{m+1}(\rho_q)+{\cal{R}}_{m-1}(\rho_q)=(2+\rho_q){\cal{R}}_{m}(\rho_q),\ \forall m,q\in\mathbb{N}$ (and ${\cal{R}}_{-1}(\rho_q)=0$), where we have set for convenience $\rho_q=-\lambda_q$, and ${\cal{R}}_{m}(\rho_q):={\cal{R}}_{mq}$. It is the $x=\rho_q$ evaluation of
\begin{equation}
{\cal{R}}_{m+1}(x)+{\cal{R}}_{m-1}(x)=(2+x){\cal{R}}_{m}(x),\ \forall m\ge1,\ {\cal{R}}_1(x)=(2+x){\cal{R}}_0(x) \label{chebyshev-gene}
\end{equation}
with ${\cal{R}}_{-1}(x)=0$. As we did for Lemma \ref{lemma-spectral}, we first truncate the operator to $N\times N$ submatrice ($0\le m,l,...\le N-1$)and set $J^N_{ml}:=(-G^N_{ml})$. Then, we can write 
\begin{equation}
J^N\cdot \left(\begin{array}{c}
\mathcal{R}_0(x)\\
\mathcal{R}_1(x)\\ 
\cdots \\
\cdots \\
\mathcal{R}_{N-1}(x)
  \end{array}  \right) + \left(\begin{array}{c}
 0\\
0\\
\cdots \\ 
0 \\
{\cal{R}}_N(x)
  \end{array}  \right)=x\left(\begin{array}{c}
\mathcal{R}_0(x)\\
\mathcal{R}_1(x)\\ 
\cdots \\
\cdots \\
\mathcal{R}_{N-1}(x)
  \end{array}  \right)\label{finitematrixrecur},
\end{equation}
where (setting $\mathcal{R}_0(x)=f(x)$) ${\cal{R}}_m(x)=f(x)U_m({{2+x}\over{2}})$, $0\le m\le N$. This leads immediately to the proof of ii) of Proposition \ref{kinetic-G}. We define $\lambda_{N,k}:=2\left(\cos({{(k+1)\pi}\over{N+1}})-1\right), \ k\in\{0,1,...,N-1\}, N\in\mathbb{N}^*$ (see \eqref{spectrum1}). On the other hand, one readily realizes from \eqref{finitematrixrecur} that any vector of $\ell^2(\mathbb{N})$ of the form 
\begin{equation}
v_{N,k}=\sum_{p=0}^{N-1}{\cal{R}}_{p}(\lambda_{N,k})e_p=\sum_{p=0}^{N-1}f(\lambda_{N,k})U_p\left({{2+\lambda_{N,k}}\over{2}}\right)e_p\label{eigenvectorG}
\end{equation}
is an eigenvector for the operator $G\in{\cal{B}}(\ell^2(\mathbb{N}))$. Hence the point spectrum of $G$ is given by the spectrum of $G^N$ \eqref{spectrum1}. Now from \eqref{unitarytrans} and the expression for the $\mathcal{R}_n(x)$'s, we obtain 
\begin{equation}
\sum_{p=0}^{N-1}{\cal{R}}_p(\lambda_{N,m}){\cal{R}}_p(\lambda_{N,l})=f(\lambda_{N,m})f(\lambda_{N,l})\sum_{p=0}^{N-1}
U_p\left(\frac{2+\lambda_{N,m}}{2}\right)U_p\left(\frac{2+\lambda_{N,m}}{2}\right)\label{checkorthog1}.
\end{equation}
This relation can be calculated by using the Christoffel-Darboux formula of Proposition \ref{chris-darb-rel}. \eqref{checkorthog1} automatically vanishes whenever $m\ne l$ since it appears only terms involving $U_N({{2+\lambda_{N,k}}\over{2}})$ which are therefore equal to zero since the ${{2+\lambda_{N,k}}\over{2}}$'s are the roots of $U_N$. When $m=l$, one has, setting $t^N_k={{2+\lambda_{N,k}}\over{2}}$,
\begin{equation}
\sum_{p=0}^{N-1}(U_p(t^N_m))^2=U_N^\prime(t^N_m)U_{N-1}(t^N_m)={{N+1}\over{((t^N_m)^2-1)}}(T_{N+1}(t^N_m)U_{N-1}(t^N_m))\label{relat-interm2},
\end{equation}
where we used $U^\prime_N(x)={{(N+1)T_{N+1}(x)-xU_N(x) }\over{x^2-1 }}$ in which $T_N(x)$ is the $N$-th order Chebyshev polynomial of 1st kind, $T_N(\cos\theta)=\cos(N\theta)$). Standard calculations combining \eqref{relat-interm2} to \eqref{unitarytrans} and \eqref{checkorthog1} yields finally
\begin{equation}
f(\lambda_{N,m})=\bigg((-1)^m(N+1){{ \sin[{{N(m+1)\pi }\over{N+1 }}  ]}\over {\sin^3[{{(m+1)\pi }\over{N+1 }} ] }} \bigg)^{-{{1}\over{2}}},\ 0\le p,m\le N-1\label{finalrml}.
\end{equation}
Observes finally that \eqref{checkorthog1}-\eqref{finalrml} yield the construction of unit vectors. Hence, $\{v^N_k\}$ \eqref{eigenvectorG} is an orthonormal family. This completes the proof for i). \\
Now, by recalling that the spectrum of a self-adjoint operator $T$ involves only (generalized) eigenvalues (i.e eigenvalues or limit point values $\lambda\in\mathbb{R}$ such that there exists a sequence of unit vectors in $Dom(T)$ $\{f_n\}_{n\in\mathbb{N}}$ such that $\lim_{n\to\infty}(T-\lambda\bbone)f_n=0$). (See e.g Theorem 2 p.170 of \cite{helmberg}). This together with i) completes the proof of ii).
\end{proof}
\begin{remark} \ \\[-20pt]
\begin{itemize}
\setlength{\itemsep}{-1pt}
\item i) We will show below that $G$ is invertible (see Lemma \ref{propagator-spectform}). Then, the orthonormal family $\{v^N_k \}$ \eqref{eigenvectorG} is an orthonormal basis for $\ell^2(\mathbb{N})$.
\item ii) It is instructive to make further comments on the full spectrum of $G$ \eqref{spectrum-full}. We have separated the point (discrete) part of the spectrum from the limit value $0$. We now re-examine the point ii) of Lemma \ref{lemma-spectrum-full} without using the theorem mentionned in the above proof. Let us first show that $G$ has no continuous spectrum, except possibly $0$ that will be examined below.
\item In view of point iii) of Remark \ref{locate-spectrum}, we show that any real value $\lambda\in]0,||G||]$, $\lambda\ne\lambda^N_k$, $\forall k\in\{0,1,2,...,N-1 \}$, $N\in\mathbb{N}, N\ne0$ is a regular value for the operator $G$. Indeed, one has $(G-\lambda\bbone)v^N_k=-(\lambda-\lambda^N_k)v^N_k$, with $|\lambda-\lambda^N_k|$ bounded, independently of $N$ and $k$. Then, any $v\in\ell^2(\mathbb{N})$ can be written as $v=\sum_{N,k}\alpha^N_kv^N_k$ with $\sum_{N,k}|\alpha^N_k|^2<+\infty$; thus $(G-\lambda\bbone)v=-\sum_{\alpha^N_k}(\lambda-\lambda^N_k)v^N_k$ with $\sum_{N,k}|\alpha^N_k(\lambda-\lambda^N_k)
|^2<+\infty$ still holds true since stemming from the above boundedeness of $|\lambda-\lambda^N_k|$. Then $(G-\lambda\bbone)v\in\ell^2(\mathbb{N})$. Hence, $\lambda$ is a regular value. To determine if $0$ belongs or not to the continuous spectrum, we simply consider the resolvent operator which for $\lambda=0$ is simply $G^{-1}$ which is unbounded since $0$ is in the spectrum. Hence, one can identify $\{0\}$ with the continuous spectrum which is reduced to a unique point in the present case.
\end{itemize}
\end{remark}
\begin{lemma}\label{propagator-spectform}
The inverse of the kinetic operator $G$ as given in Proposition \ref{kinetic-G} is defined by the following unbounded operator on $\ell^2(\mathbb{N})$
\begin{eqnarray}
P_{ml}&=&{{1}\over{\pi\mu^2}}\int_{-1}^1dx\ {\sqrt{{{1+x }\over{ 1-x}}}}U_m(x)U_l(x)\label{invers-G}.
\end{eqnarray}
\end{lemma}
\begin{proof}
This is a direct application of the discussion given at the end of section 2. For further convenience, it is better to express the propagator with the ''standard'' Chebyshev polynomials. The obtention of \eqref{invers-G} is just routine calculation that uses the recurrence for the Chebyshev polynomials $U_m$. Notice that one can easily check for instance $\sum_{l\in\mathbb{N}}G_{ml}P_{lr}={{1}\over{\pi}}\int_{-1}^1dx{\sqrt{\frac{1+x}{1-x}}}(2\delta_{ml}-\delta_{m,l+1}-\delta_{m,l-1})U_l(x)U_r(x)$. Then, use $U_{n+1}(x)+U_{n-1}(x)=2xU_n(x)$ to put the integral into the form $\sum_{l\in\mathbb{N}}G_{ml}P_{lr}=\frac{2}{\pi}\int_{-1}^1dx{\sqrt{\frac{1+x}{1-x}}}(1-x)U_l(x)U_r(x)$ and finally \eqref{orthocontinucheb2} to show that $P_{ml}$ is the right inverse of $G$. In the same way $\sum_{l\in\mathbb{N}}P_{ml}G_{lr}=\delta_{mr}$.
\end{proof}

\begin{remark}\label{remark-dirac} \ \\[-20pt]
\begin{itemize}
\setlength{\itemsep}{-1pt}
\item i) As an illustration of the discussion at the end of the section 2, it is easily observed that $G_{ml}=\langle e_m,Ge_l \rangle $ \eqref{tridiagon} can be cast into the form 
\begin{equation}
G_{ml}=\int_{-4}^0t\ (P_m(t)P_l(t)d\mu(t)),\ \ d\mu(t)=-\frac{1}{\pi}dt(-\frac{t}{2}(2+\frac{t}{2}))^{\frac{1}{2}}\label{spectral-kineticterm}
\end{equation}
where the polynomials $P_n$ satisfy $P_{n+1}(t)+P_{n-1}(t)-2P_n(t)=tP_n(t)$, $P_1(t)-2P_0(t)=tP_0(t)$, $P_0(t)=1$ (see e.g the proof of Proposition \ref{kinetic-G}). This indeed agrees with the Favard Theorem \ref{Favard} and especially the spectral theorem \ref{orthop-1} (recall eqn. \eqref{spectral-kinetic-operat}).
\item ii) From the above analysis, we define the following family of projectors
\begin{equation}
P^N_k:\ell^2(\mathbb{N})\to {\cal{E}}^N_k:=span\{v^N_k\},\ P^N_k:=\lambda^N_k| v^N_k\rangle\langle v^N_k |,\ k\in\{0,1,...N-1 \}, N\in\mathbb{N}^*\label{proj-simple}.
\end{equation}
By further combining this spectral family with Propositions \ref{kinetic-G} and \ref{lemma-diag-1}, one infers
\begin{equation}
G^N=\sum_{n=1}^N\sum_{k=0}^{n-1}\lambda^n_kP^n_k\label{gn-projector},
\end{equation}
and one easily realizes that $\lim_{N\to\infty}G^N=G$ where the convergence holds true for the strong operator topology. Indeed, for any $f\in\ell^2(\mathbb{N})$, $f=\sum_{N,k}f^N_kv^N_k$, a simple computation yields $$||(G^L-G)f ||^2_2=\sum_{n=L}^\infty\sum_{k=0}^{n-1}|\lambda^n_k|^2|f^n_k|^2\le\max\{|\lambda^n_k|^2 \}\sum_{n=L}^\infty\sum_{k=0}^{n-1}|f^n_k|^2\le \max\{|\lambda^n_k|^2 \}||f||_2^2$$. 
\noindent Set $S_N:=\sum_{n=1}^N\sum_{k=0}^{n-1}|f^n_k|^2$. Obviously, for any $\varepsilon>0$, one can find $N_0\in\mathbb{N}$ such that for any $n>N_0$, one has $|S_n-||f||^2_2|<\varepsilon$. Now, simply write $|\sum_{n=L}^N\sum_{k=0}^{n-1}|f^n_k|^2|=|S_N-S_L|=|S_N-||f||^2_2|-S_L+||f||^2_2|$ and uses the convergence condition for $S_N$ to show that $\lim_{L\to\infty}||(G^L-G)f ||^2_2=0$ for any $f\in\ell^2(\mathbb{N})$. Hence $\lim_{N\to\infty}G^N=G$.
\end{itemize}
\end{remark}
Keeping in mind the point i) of Remark \ref{locate-spectrum} and the point ii) above, we define the following self-adjoint operator $D\in{\cal{B}}(\ell^2(\mathbb{N}))$ such that $D^2=G$
\begin{equation}
D:=\sum_{n=1}^\infty\sum_{k=0}^{n-1}\alpha^n_kP^n_k, \ (\alpha^n_k)^2=\lambda^n_k,\ k\in\{0,1,...N-1 \}, N\in\mathbb{N}^*\label{Dirac-simple},
\end{equation}
with $\ker(D)=\{0 \}$, which will be used below as a Dirac operator involved in a spectral triple.


\section{Dirac operators and spectral triples}\label{ncg}
We now examine if a spectral triple built from the above Dirac operator can support additional structures such as regularity or summability properties. We have of course in mind a spectral triple with the square root of the kinetic operator as Dirac operator from which the NCFT action \eqref{actmatrix} would be obtained as a spectral action. This would correspond to the noncommutative geometry underlying the gauge-fixed action \eqref{actmatrix}.\par 

Some properties of the related spectral triple do not depends on the explicit expression of $D$ but only on the classes of operators it belongs to. Here, one salient operatorial property of $D$ \eqref{Dirac-simple} is that it is bounded. We will therefore consider first the general case of an unspecified self-adjoint operator ${\cal{D}}$ with $Dom({\cal{D}})=\ell^2(\mathbb{N})$ (hence automatically bounded). This together with the fact that the relevant $*$-representation $\eta$ defined in Section \ref{alg2} 
\begin{equation}
\eta:{\cal{M}}_\theta\to{\cal{B}}(\ell^2(\mathbb{N})),\ \eta(e_m\otimes e_n)=e_m\otimes e^*_n,\ \forall m,n\in\mathbb{N}\label{Left-representation}
\end{equation}
maps elements of ${\cal{M}}_\theta$ into Hilbert-Schmidt operators on $\ell^2(\mathbb{N})$ will single out a limited class of spectral triples. We will then specialize to the Dirac operator \eqref{Dirac-simple} when necessary.\par

We denote as usual by ${\cal{L}}^p({\cal{H}})$, $p\ge1$, the $p$-th Schatten ideals of ${\cal{B}}({\cal{H}})$ (the Hilbert space ${\cal{H}}$ will be set equal to $\ell^2(\mathbb{N})$ in a while). Schatten norm on ${\cal{L}}^p({\cal{H}})$ 
is denoted by $||.||_p$. It is convenient to recall the definition of a spectral triple that will be used in the sequel.

\begin{definition}\label{spect-triple}
A spectral triple is the set of data $(\mathbb{A},\pi,{\cal{H}},{\cal{D}})$ where $\mathbb{A}$ is an involutive ($C^*$-)algebra and ${\cal{H}}$ is a separable Hilbert space carrying a faithfull $*$-representation of $\mathbb{A}$ and  ${\cal{D}}$ a self-adjoint densely defined operator on ${\cal{H}}$, non necessarily bounded, such that for any $a\in\mathbb{A}$:
\begin{itemize}
\vspace*{-3pt}
\setlength{\itemsep}{-1pt}
\item i) $\pi(a)$ Dom$({\cal{D}})\subseteq$ Dom$({\cal{D}})$,
\item ii) $[{\cal{D}},\pi(a)]\in{\cal{B}}({\cal{H}})$,
\item iii) for any $z\notin spec({\cal{D}})$, $\pi(a)R_{{\cal{D}}}(z)$ is a compact operator on ${\cal{H}}$, where $R_{{\cal{D}}}$ is the resolvent operator of ${\cal{D}}$.
\end{itemize}
\end{definition}

\begin{definition}\label{def-regular}
A spectral triple is regular if (in obvious notations) $\pi(\mathbb{A})\cup[{\cal{D}},\pi(\mathbb{A})]\subseteq\bigcap_{j\in\mathbb{N}}Dom(\delta^j)$ where $\delta T:=[|{\cal{D}}|,T]$ for any $T\in{\cal{B}}({\cal{H}})$.
\end{definition}
Notice that Definition \ref{spect-triple} corresponds to an odd spectral triple. We will introduce additional simple structures in a while to built even triples.\par

We now specialize to the case ${\cal{H}}=\ell^2(\mathbb{N})$ and $\pi=\eta$ \eqref{Left-representation} but still do not focus on some specific Dirac operator, only assuming that it is self-adjoint with $Dom({\cal{D}})=\ell^2(\mathbb{N})$. We can characterize interesting properties of the spectral triple built from ${\cal{D}}$ by the following small theorem. To simplify the notation, we now set 
${\cal{L}}^p(\ell^2(\mathbb{N}))={\cal{L}}^p$.
\begin{lemma}\label{th1}
The data $\mathfrak{X}_{{\cal{D}}}:=({\cal{M}}_\theta,\eta,\ell^2(\mathbb{N}), {\cal{D}}))$ is a regular spectral triple.
\end{lemma}
\begin{proof}
To prove that $\mathfrak{X}_{{\cal{D}}}$ is a spectral triple, we first need to use simple properties stemming from the fact that ${\cal{D}}$ is bounded and $\eta(a)$ is Hilbert-Schmidt, $\eta(a)\in{\cal{L}}^2$. Let $R_{\cal{D}}(z)$ denotes the resolvent operator for ${\cal{D}}$. In fact, for any $a\in{\cal{M}}_\theta$ and $z\notin spec({\cal{D}})$, the following properties hold true: i) $[{\cal{D}},\eta(a)]\in{\cal{L}}^2$, ii) $\eta(a)R_{\cal{D}}(z)\in{\cal{L}}^2$, iii) $[\eta(a),R_{\cal{D}}(z)]\in{\cal{L}}^2$. Points i), ii) and iii) are direct consequences of the fact that ${\cal{L}}^p$ is a 2-sided ideal of ${\cal{B}}(\ell^2(\mathbb{N}))$. Anyway, it is instructive to give a less direct proof by using H\"older estimates. Namely, use $||T||\le||T||_2$ and $||ATB||_p\le||A||||T||_p||B||$ for any $T\in{\cal{L}}^p$, $A,B\in{\cal{B}}(\ell^2(\mathbb{N}))$ to obtain for any $a\in{\cal{N}}_\theta$
\begin{equation}
||[{\cal{D}},\eta(a)]||\le||[{\cal{D}},\eta(a)]||_2\le2||{\cal{D}}||||\eta(a)||_2<+\infty\label{estimate-triple11}.
\end{equation}
Hence i) is proven. Next, for any $a\in{\cal{M}}_\theta$ and any $z\notin spec({\cal{D}})$, the following estimate holds true
\begin{equation}
||\eta(a)R_{\cal{D}}(z)||_2\le||\eta(a)||_2||R_{\cal{D}}(z)||<\infty, ||[\eta(a),R_{\cal{D}}(z)]||_2\le2||R_{\cal{D}}(z)||||\eta(a)||_2<\infty\label{estimate-triple22}
\end{equation}
so that $\eta(a)R_{{\cal{D}}}(z)\in{\cal{L}}^2$, $[\eta(a),R_{\cal{D}}(z)]\in{\cal{L}}^2$ and ii) and iii) are true.\\
Now, ${\cal{M}}_\theta$ is non unital involutive, ${\cal{D}}$ self-adjoint defined everywhere on $\ell^2(\mathbb{N})$ by assumption and the $*$-representation $\eta$ \eqref{Left-representation} is faithfull by construction. Besides, property i) implies that $[{\cal{D}},\eta(a)]\in{\cal{B}}(\ell^2(\mathbb{N}))$ while property ii) implies that for any $z\notin spec({\cal{D}})$, $\eta(a)R_{\cal{D}}(z)$ is a compact operator on $\ell^2(\mathbb{N})$ and one has obviously $\eta(a)v\in\ell^2(\mathbb{N})$ for any $a\in{\cal{M}}_\theta$. Hence, $\mathfrak{X}_{{\cal{D}}}$ is a spectral triple.\\
Let $\delta:=[|{\cal{D}}|,.]\in Der({\cal{B}}(\ell^2(\mathbb{N}))) $. Since $|{\cal{D}}|\in{\cal{B}}(\ell^2(\mathbb{N}))$, one has $Dom(\delta)={\cal{B}}(\ell^2(\mathbb{N}))$ and by simple induction $Dom(\delta^j)={\cal{B}}(\ell^2(\mathbb{N}))$ for any $j\in\mathbb{N}$. Therefore
\begin{equation}
\bigcap_{j\in\mathbb{N}}Dom(\delta^j)={\cal{B}}(\ell^2(\mathbb{N}))\label{rhs-regular}.
\end{equation}
But one has $\eta({\cal{M}}_\theta)\subseteq{\cal{L}}^2$ and $[{\cal{D}},\pi({\cal{M}}_\theta)]\subseteq{\cal{B}}(\ell^2(\mathbb{N}))$ so that Definition \ref{def-regular} is verified. Hence, the spectral triple $\mathfrak{X}_{{\cal{D}}}$ is regular.\\
\end{proof}

\noindent At this point, two remarks are in order

\begin{remark}
Since ${\cal{D}}$ is bounded together with the fact that we use the natural Hilbert-Schmidt action of ${\cal{M}}_\theta$ on $\ell^2(\mathbb{N})$ throught the representation $\eta$, it can be expected that the above spectral triple $\mathfrak{X}_{{\cal{D}}}$ is, informally speaking, ''close to'' a Fredholm module. Recall that a Fredholm module, i.e an analytic $K$-cycle, is defined as the set of data $(\mathbb{A},\pi,{\cal{H}},F)$ as in Definition \ref{spect-triple} with $\pi$ a $*$-representation and the operator $F\in{\cal{B}}({\cal{H}})$ is such that for any $a\in\mathbb{A}$, $\pi(a)(F^2-1)$, $\pi(a)(F-F^*)$ and $[\pi(a),F]$ are compact operators on ${\cal{H}}$. Furthermore, a Fredholm module is $p$-summable if the $*-$algebra $\mathfrak{A}:=\{a\in\mathbb{A}: [\pi(a),F]\in{\cal{L}}^p({\cal{H}}),\ \pi(a)(F-F^*)\in{\cal{L}}^{\frac{p}{2}}({\cal{H}}),\ \pi(a)(F^2-1)\in{\cal{L}}^{\frac{p}{2}}({\cal{H}})\}$ is norm dense in $\mathbb{A}$, for non zero $p\in\mathbb{N}$. In the present case, it can be easily realized that the data ${\cal{F}}:=({\cal{M}}_\theta,\eta,\ell^2(\mathbb{N}), F)$ with $F:=\frac{{\cal{D}}}{|{\cal{D}}|}$ together with $F:=0$ on $\ker({\cal{D}}$) and ${\cal{D}}$ (assuming ${\cal{D}}$ is invertible) define a Fredholm module which is 2-summable.
\end{remark}

\begin{remark}
From now on, we focus on the Dirac operator \eqref{Dirac-simple}. For further use, we define the isometry
\begin{equation}
\mathfrak{J}:\ell^2(\mathbb{N})\to\ell^2(\mathbb{N}),\ \mathfrak{J}:=\sum_{n,k}u^n_kP^n_k,\ (u^n_k)^2=1,\ \forall k\in\{0,1,...,n-1\},\ n\in\mathbb{N}^*\label{isomet-basic},
\end{equation}
where the projectors $P^n_k$ are given by \eqref{proj-simple}. $\mathfrak{J}$ satisfies
\begin{equation}
\mathfrak{J}^2=\bbone,\ D\mathfrak{J}=\mathfrak{J}D\label{commut-1}. 
\end{equation}
According to the analysis of Section 3, any arbitrary element of ${\cal{N}}_\theta$ can be written as 
\begin{equation}
a=\sum a_{n_1k_1;n_2k_2}v^{n_1}_{k_1}\otimes v^{n_2}_{k_2}
\end{equation}
(with $\sum |a_{n_1k_1;n_2k_2} |^2<\infty$) where the $v^{n_i}_{k_i}$'s define the orthonormal basis of $\ell^2(\mathbb{N})$ determined by Lemma \ref{lemma-spectrum-full}. Then, it can be easily realized (in view of e.g \eqref{isomet-basic}) that any element of ${\cal{M}}_\theta$ of the form 
\begin{equation}
a_c=\sum a_{nk}v^n_k\otimes v^n_k 
\end{equation}
commutes with $D$ \eqref{Dirac-simple}. This shows that the spectral triple $\mathfrak{X}_D$, theorem \ref{th1}, related to $D$ does not fullfill one of the necessary conditions for a spectral triple to define a metric space with the Connes spectral distance \cite{Rieffel}. This is unlike the Moyal metric space and its homothetic descendants which all verify this necessary condition. Recall that the latter metric commutant condition is 
\begin{equation}
[D,\eta(a)]=0\iff a=\lambda\bbone,\ \lambda\in\mathbb{C},
\end{equation}
which is not verified here. In other words, $\mathfrak{X}_D$ is not a spectral metric space in the sense of \cite{homot-moyal,bel-mar}. Despite this lost of metric structure for the Connes spectral distance, it turns out that the spectral triple $\mathfrak{X}_D$ can be enlarged with additional algebraic structures, as we now show.
\end{remark}

Let us introduce the Pauli matrices
\begin{equation}
\sigma_1=\begin{pmatrix}
0&1\\ 
1&0
\end{pmatrix}, \sigma_2=\begin{pmatrix}
0&i\\ 
-i&0
\end{pmatrix},\ \sigma_3=i\sigma_1\sigma_2=\begin{pmatrix}
1&0\\
0&-1\end{pmatrix},
\end{equation}
which satisfy $\{\sigma_\mu,\sigma_\nu\}=2\delta_{\mu\nu}\bbone_2$, $\mu,\nu=1,2$ and define an irreducible representation of the Clifford algebra $\mathbb{C}l(2)$. We also define:
\begin{equation}
\gamma^\mu:=\bbone_2\otimes\sigma_\mu,\ \gamma^{\mu+2}:=\sigma_\mu\otimes\bbone_2,\ \mu,\nu=1,2,\label{ko-clifford}
\end{equation}
such that $\{\gamma_\mu,\gamma_\nu \}=2\delta_{\mu\nu}\bbone_4$, $\{\gamma_{\mu+2},\gamma_{\nu+2} \}=2\delta_{\mu\nu}\bbone_4$, $\mu,\nu=1,2$. Let us define the Hilbert space as ${\cal{H}}_0:=\ell^2(\mathbb{N})\otimes\mathbb{C}^4$. We will need to define a real structure and a grading to get an even triple. Accordingly, we introduce the self-adjoint grading operator $\Gamma$ (i.e defining the chirality) and unitary antilinear operator (i.e charge conjugation) ${\cal{J}}$ to implement reality condition. They are defined by:
\begin{equation}
\Gamma:=\gamma_1\gamma_2\gamma_3\gamma_4=-\sigma_3\otimes\sigma_3,\ \Gamma^2=\bbone_4\label{chiral},
\end{equation}
\begin{equation}
{\cal{J}}:=\gamma_2{\cal{C}}=(\bbone_2\otimes\sigma_2){\cal{C}},
\end{equation}
where ${\cal{C}}$ is the complex conjugation. We next pick the following Dirac operator
\begin{equation}
{\cal{D}}:=D\gamma_3=D\sigma_1\otimes\bbone_2,
\end{equation}
with $D$ given in \eqref{Dirac-simple}. We quote a useful set of algebraic properties.
\begin{proposition}\label{KO-real}
The following properties holds true:
\begin{eqnarray}
{\cal{J}}^2&=&-\bbone_2,\ {\cal{J}}{\cal{D}}={\cal{D}}{\cal{J}},\ {\cal{J}}\Gamma=-\Gamma{\cal{J}},\label{KO-dim}\\ 
\Gamma^2&=&1,\ {\cal{D}}\Gamma=-\Gamma{\cal{D}}\label{real-struct-grad}.
\end{eqnarray}
\end{proposition}
\begin{proof}
This is routine computation.
\end{proof}
Now it is not difficult to equip the spectral triple of lemme \ref{th1} with additional structures. Namely, one has
\begin{theorem}\label{th2}
The following data:
 \begin{equation}
({\cal{M}}_\theta,\pi:=\eta\otimes\bbone_4,{\cal{H}}_0:=\ell^2(\mathbb{N})\otimes\mathbb{C}^4,{\cal{D}}:=D\gamma_3;{\cal{J}}:=\gamma_2{\cal{C}},\Gamma:=-\sigma_3\otimes\sigma_3),\label{ncg-fluctuat}
\end{equation}
where $D$ is given by \eqref{Dirac-simple}, is a regular spectral. The triple is even, supports a weak real structure defined by ${\cal{J}}$ and $\Gamma$ with KO-dimension equal to 2 and verifies the commutant condition:
\begin{itemize}
\vspace*{-4pt}
\setlength{\itemsep}{-1pt}
\item i) $[\pi(a),{\cal{J}}^{-1}\pi(b^\dag){\cal{J}}]=0$, $\forall a,b\in{\cal{M}}_\theta$,\\
and the first order condition modulo compact operators;
\item ii) $[[\pi(a),D],{\cal{J}}^{-1}\pi(b^\dag){\cal{J}}]\in{\cal{K}}({\cal{H}}_0)$, $\forall a,b\in{\cal{M}}_\theta$.
\end{itemize}
\end{theorem}
\begin{proof}
The bounded self-adjoint operator ${\cal{D}}$ has obviously its domain $Dom({\cal{D}})={\cal{H}}_0$, while the $*$-representation $\pi$ inherits faithfullness from $\eta$. One has $\pi(a)Dom({\cal{D}})\subseteq {\cal{H}}_0=Dom({\cal{D}})$ for any $a\in{\cal{M}}_\theta$, thanks to $\eta(a)\in{\cal{L}}^2$ and the diagonal form of $\pi(a):=\eta(a)\otimes\bbone_4$. Next, one checks that for any $a\in{\cal{M}}_\theta$
\begin{eqnarray}
[{\cal{D}},\pi(a)]&=&[D,\eta(a)]\sigma_1\otimes\bbone_2,\label{commut-tensor},\\
\pi(a)({\cal{D}}^2+1)^{-1}&=&\eta(a)(D^2+1)^{-1}\otimes\bbone_4\label{resol-tensor}.
\end{eqnarray}
Then, one has $[D,\eta(a)]\in{\cal{B}}(\ell^2(\mathbb{N}))$ which implies $[{\cal{D}},\pi(a)]\in{\cal{B}}({\cal{H}}_0)$ thanks to the action on ${\cal{H}}_0$ of this operator, see \eqref{commut-tensor}. Indeed, set $\ell(a):=[{\cal{D}},\pi(a)]$. Then, for any $f\in{\cal{H}}_0$, $\langle \ell(a)f ,\ell(a)f \rangle=\langle f ,\ell(a)^\dag\ell(a)f \rangle=\langle f ,[D,\eta(a)]^\dag[D,\eta(a)]\otimes\bbone_4f \rangle=\sum_i\langle [D,\eta(a)]f_i , [D,\eta(a)]f_1\rangle$, where $f=(f_i)_{i=1,...,4}$ and the RHS of the last equality is bounded because $[D,\eta(a)]$ is a bounded operator. \\
In the same way, $\eta(a)(D^2+1)^{-1}\in{\cal{L}}^2$ implies that $\pi(a)({\cal{D}}^2+1)^{-1}\in{\cal{L}}^2({\cal{H}}_0)$ (in view of its diagonal action, see \eqref{resol-tensor}), hence it is compact. Thus, from the Definition \ref{spect-triple}, $({\cal{N}}_\theta,\pi:=\eta\otimes\bbone_4,{\cal{H}}_0:=\ell^2(\mathbb{N})\otimes\mathbb{C}^4,{\cal{D}}:=D\gamma_3)$ is a spectral triple.\\
Next, $|{\cal{D}}|$ is bounded and $\delta:=[|{\cal{D}}|,.]\in Der({\cal{B}}({\cal{H}}_0))$ has its domain $Dom(\delta)={\cal{B}}({\cal{H}}_0)$. As for Lemma \ref{th1}, one has $\cap_{j\in\mathbb{N}}Dom(\delta^j) ={\cal{B}}({\cal{H}}_0)$. Besides, $\pi({\cal{M}}_\theta)\subseteq{\cal{L}}^2({\cal{H}}_0)$ and $[{\cal{D}},\pi({\cal{M}}_\theta)]\subseteq{\cal{B}}({\cal{H}}_0)$. Hence $\pi({\cal{M}}_\theta)[{\cal{D}},\pi({\cal{M}}_\theta)]\subseteq\cap_{j\in\mathbb{N}}Dom(\delta^j)$ and the triple is regular.\\
The algebraic properties, Proposition \ref{KO-real}, imply immediately that the triple is even with weak real structure (see Remark \ref{weak-struc} below) defined by ${\cal{J}}$ and $\Gamma$ and with KO-dimension equal to 2, which can be read off from e.g the table of 2nd of ref \cite{Connes1} (pp. 192, Definition 1.124).\\
We have also to consider the conditions of commutant and first order. As mere consequences of ${\cal{J}}\in{\cal{B}}({\cal{H}}_0)$, $\pi(a)\in{\cal{L}}^2({\cal{H}}_0)$ and the fact that ${\cal{L}}^2({\cal{H}}_0)\subset{\cal{K}}({\cal{H}}_0)$ is a two-sided ideal of ${\cal{B}}({\cal{H}}_0)$, one has immediately
\begin{equation}
[\pi(a),{\cal{J}}^{-1}\pi(b^\dag){\cal{J}}]\in{\cal{K}}({\cal{H}}_0),\ [[\pi(a),D],{\cal{J}}^{-1}\pi(b^\dag){\cal{J}}]\in{\cal{K}}({\cal{H}}_0),\ \forall 
a,b\in{\cal{M}}_\theta.
\end{equation}
Therefore, the point ii) of the theorem is verified. Let us compute the LHS of the first relation. One has for any $f=(f_i)_{i=1,...,4}\in{\cal{H}}_0$ and $b\in{\cal{M}}_\theta$,  ${\cal{J}}^{-1}\pi(b^\dag){\cal{J}}f=(R(b)\otimes\bbone_4)f$ where $R$ is the right multiplication. But one can also write $\pi(a) f=(L(a)\otimes\bbone_4)f$ where $L$ is the left multiplication which verifies $[L(a)\otimes\bbone_4,R(b)\otimes\bbone_4]=0$ for any $a,b\in{\cal{M}}_\theta$. Hence the commutant condition i) holds true.
\end{proof}
\begin{remark}\label{weak-struc} At this point, some comments are in order:
\begin{itemize}
\vspace*{-4pt}
\setlength{\itemsep}{-1pt}
\item i) It is worth pointing out that we consider here a weak notion of real structure, as the one used e.g in spectral triples on quantum groups for which commutant condition and first order condition are fullfilled only modulo a 2-sided ideal in compact operators \cite{weak-real}.
\item ii) The Dirac operator used in Theorem \ref{th2} does not satisfies the Leibnitz rule. This latter , if it would have been satisfied, would have implied that the first order condition is verified, $[[\pi(a),D],{\cal{J}}^{-1}\pi(b^\dag){\cal{J}}]=0$.
\item iii) Notice that Proposition \ref{KO-real} still holds true with ${\cal{J}}$ replaced by ${\cal{J}}^\prime=\gamma_2\mathfrak{J}{\cal{C}} $ where $\mathfrak{J}$ given by \eqref{isomet-basic}. This provides another weak real structure for the spectral triple of 
Theorem \ref{th2} with ${\cal{J}}$ replaced by ${\cal{J}}^\prime$, with the same properties, except the commutant property, point i) of Theorem \ref{th2}, which is verified only modulo compact operators.
\end{itemize}
\end{remark}


\section{Discussion}\label{discuss}

Bounded or unbounded Jacobi operators defined by 3-term recurrence are usually associated with families of hypergeometric orthogonal polynomials pertaining to the Askey scheme. These operators appear as kinetic operators in NCFT when expressed as ''matrix models'' within some matrix base formalism, according to the usual liturgy of NCFT, provided some law of indice conservation holds true. This is the case of NCFT on Moyal space and Moyal plane, e.g the Grosse-Wulkenhaar model and its extensions and the scalar NCFT and gauge theory models on $\mathbb{R}^3_\lambda$. Thus, the case of kinetic Jacobi operator covers numerous NCFT of interest. \par
In this paper, we have only considered the case of bounded Jacobi operators for which interesting conclusions can already be drawn. The unbounded case will be considered elsewhere \cite{unboud-jac}. As a representative example, we have analyzed in detail, using the mathematical toolkit developed in the section 2, a particular noncommutative gauge theory model on the Moyal plane introduced in \cite{MVW13}. It is obtained from the so called induced gauge theory of \cite{GWW,GW07} expanded around a non trivial symmetric vacuum and suitably gauge-fixed. For related details, \cite{MVW13}. By the way, a short review on the present situation of noncommutative gauge theories is given in subsection \ref{section-review}. The resulting kinetic operator $G$ is bounded Jacobi, related to Chebyshev polynomials of 2nd order.\par 
A simple application of Section 2 to the above NCFT yields a complete characterization of the (spectral) properties of the kinetic operator, its full spectrum and the corresponding propagator, therefore improving and clarifying the results of \cite{MVW13}. This combined with general considerations permits one to obtain robust conclusions on the expected perturbative behavior of the NCFT as well as on information on its physical properties that we now discuss. First, the UV and IR regions are naturally identified respectively with large and small indices, says $m,n,...\gg 1$ or $m,n,... =0,1$, by simply rescaling the gauge-fixed action by a factor $g^2$ with mass dimension $[g]=1$ (hence $[\mu]=1$), which then from \eqref{spectrum-full} correspond respectively to large or small eigenvalues for the ''Laplacian'' $G$.
From Lemma \ref{lemma-spectrum-full}, one sees that $G$ has $0$ as a limit point in its spectrum which can be interpreted as a kinetic operator of a ''massless'' theory. Accordingly, the corresponding propagator Lemma \ref{propagator-spectform} is unbounded. At the perturbative level, one thus expect that this NCFT will exhibit - correlations at large separation between indices, which can be realized at the level of the eigenvalues of the propagator which do not decay in the UV region. This is unlike, e.g the Grosse-Wulkenhaar model \cite{gw1,gw2} or the scalar of gauge NCFT on $\mathbb{R}^3_\lambda$ \cite{gvw13} characterized by propagators with sufficient UV damping and suppression of correlations at large indices separation.\par
Such a behavior for the above gauge model renders a perturbative treatment questionable and simply reflects the general fact that a bounded (kinetic) operator cannot give rise to a compact propagator (i.e inverse) with a decaying behavior (in e.g some UV region). As far as the Askey-scheme is concerned, bounded Jacobi operators correspond to the serie of orthogonal Jacobi polynomials and descendants (Jacobi, Gegenbauer, Chebyshev of type 1 and 2, spherical, Legendre polynomials). Any related NCFT and/or matrix model is expected to have a similar problematic perturbative behavior. This however, does not forbid some among these models to be solvable.\par

The cubic vertex in the gauge-fixed action triggers the appearance at the one-loop order of a non-vanishing one-point (tadpole) function. The corresponding computation is given in the \ref{onepoint-comput} for the sake of completeness. This is already the case for the gauge theory on $\mathbb{R}^3_\lambda$ considered in \cite{gvw13}. One-loop non-vanishing tadpole signals quantum vacuum instabilities and an interesting question that deserves further investigation is to examine if a noncommutative analog of the Goldstone theorem can be defined.\par

The initial action that gave rise to the NCFT studied here is the so called induced gauge theory action \eqref{inducedgaugematrix}. Recall that this later can be related to a particular NCG described by a finite volume spectral triple \cite{finite-vol} homothetic, as noncommutative metric spaces, to the standard noncommutative Moyal space \cite{homot-moyal}. However, the NCFT \eqref{actmatrix} is obtained from \eqref{inducedgaugematrix} after an expansion around some vacuum and a gauge-fixing which thus modify the initial kinetic operator in \eqref{inducedgaugematrix} and so the related Dirac operator. Then, one can expect a change in the related NCG compared to the one described by a finite volume triple. By using a Dirac operator obtained from the kinetic operator, we have shown that one can construct an even, regular spectral triple with weak real structure and KO-dimension equal to 2, obeying the commutant condition and the first order condition only modulo compact operators. This spectral triple however does not define a noncommutative metric space for the Connes spectral distance, unlike the initial NCG.\par


\vskip 2 true cm
{\bf{Acknowledgments}}: Discussions and correspondences with P. Martinetti, N. Pinamonti, H. Steinacker and P. Vitale are gratefully acknowledged. One of us (J.C.W) thanks M. Dubois-Violette for discussions on mathematical aspects related to the first order conditions and twists in noncommutative geometry.


\vskip 1 true cm

\setcounter{section}{0}
\appendix



\section{Definitions and spectral theorem for bounded operators}\label{appendix1}

For details, see e.g \cite{orthop-1}. Let ${\cal{H}}_1$ and ${\cal{H}}_2$ be 2 separable Hilbert spaces on $\mathbb{C}$ with inner product $\langle, \rangle_{{\cal{H}}_i}$, $i=1,2$. Let $T:{\cal{H}}_1\to{\cal{H}}_2$ be a {\it{linear}} operator. We recall that the operator norm is defined by $||T||:=\sup\left\{\frac{||Tu||_{{\cal{H}}_2}}{||u||_{{\cal{H}}_1}},\ u\in{\cal{H}}_1 \right\}$. The following usefull definitions that will be used in the paper are collected below.

\begin{adefinition}
T is bounded if $||T||<\infty$. The adjoint operator of a bounded operator is a linear map $T^*:{\cal{H}}_2\to{\cal{H}}_1$ satisfying $\langle u, T^*v\rangle_{{\cal{H}}_1}=\langle Tu,v \rangle_{{\cal{H}}_2}$, $\forall u\in{\cal{H}}_1,\ \forall v\in{\cal{H}}_2$. An operator $T:{\cal{H}}_1\to{\cal{H}}_1$ is 
self-adjoint if $T^*=T$; it is normal if $TT^*=T^*T$; it is unitary if $TT^*=T^*T=\bbone_{\cal{H}}$. An operator $T:{\cal{H}}_1\to{\cal{H}}_1$ is symmetric if $\langle Tu,v \rangle=\langle u,Tv\rangle$. A projection is defined as a linear operator $P:{\cal{H}}_1\to{\cal{H}}_1$ such that $P^2=P$.
\end{adefinition}

One of the main notions entering the version of the spectral theorem we use in this paper is the resolution of the identity of a Hilbert space. A convenient definition is provided below:
\begin{adefinition}\label{resolution}
A resolution of the identity $E$ of a Hilbert space ${\cal{H}}$ is a Borel measure on $\mathbb{R}$, projection valued, satisfying for any Borel set $U,\ V$ the following properties: i) $E(U)$ is a self-adjoint projection, ii) $E(U\cap V)=E(U)E(V)$, iii) $E(\emptyset)=0,\ E(\mathbb{R})=\bbone_{{\cal{H}}}$, iv) $U\cap V=\emptyset\Rightarrow E(U\cup V)=E(U)+E(V)$, v) for any $u,v\in\cal{H}$ a complex Borel measure is defined by the map $U\mapsto E_{u,v}=\langle E(U)u,v \rangle_{\cal{H}}$.
\end{adefinition}
Let ${\cal{B}}({\cal{H}})$ denotes the Banach algebra of bounded operators on ${\cal{H}}$. Recall that the resolvent set of any $T\in{\cal{B}}({\cal{H}})$, $\rho(T)$, is defined as $\rho(T):=\{z\in\mathbb{C}, R(z):=(T-z\bbone)^{-1}\in{\cal{B}}({\cal{H}})  \}$ where $R(z)$ is the resolvent operator. The spectrum of $T\in{\cal{B}}({\cal{H}})$ is defined by $spec(T):=\mathbb{C}/\rho(T)$ and is thus the subset of $\mathbb{C}$ for which $(T-z\bbone)$ has no bounded inverse. The spectrum of a bounded operator $T$ is a compact subset of the disk of radius $||T||$ while a bounded self-adjoint operator has real spectrum which is a subset of 
$[-||T||,||T||]$.\par

In this paper, we deal actually with bounded self-adjoint operators. The convenient version of the spectral theorem for these operators{\footnote{General definitions and properties for operators are recalled in the appendix \ref{appendix1}.}} that will provide the integration measures ensuring orthogonality among each family of polynomials can be stated as:
\begin{atheorem}\label{orthop-1}
Let $T:{\cal{H}}\to{\cal{H}}$ denotes a bounded self-adjoint operator. There exists a unique resolution of the identity on ${\cal{H}} $ $E$ satisfying 
\begin{equation}
\langle Tu,v \rangle_{\cal{H}}=\int_{\mathbb{R}}t\ dE_{u,v}(t)\label{spectral-kinetic-operat}
\end{equation}
and which is supported on the spectrum of $T$, $spec(T)$, verifying $spec(T)\subset[-||T|| ,||T||]$. For any Borel set $U\subset\mathbb{R}$, $E(U)$ commutes with $T$.
\end{atheorem}
The above spectral theorem yields the well known functional calculus for self-adjoint operators. It turns out that the above measure $E_{u,v}$ can be actually computed from the Stieltjes-Perron inversion formula:
\begin{atheorem}\label{stieltjes-perron}
Let $T:{\cal{H}}\to{\cal{H}}$ denotes a bounded self-adjoint operator with resolvent operator $R(z)$. For any open set $]a,b[\subset\mathbb{R}$, $u,v\in\cal{H}$, one has in the strong operator topology sense:
\begin{equation}
E_{u,v}(]a,b[)=\lim_{\varepsilon\to0}\lim_{\nu\to0 }\int_{a+\varepsilon}^{b-\varepsilon}\langle R(x+i\nu)u,v \rangle-\langle R(x-i\nu)u,v \rangle dx.
\end{equation}
\end{atheorem}


\section{One-point function computation}\label{onepoint-comput}

Adding to $S[\phi]$ \eqref{actmatrix} a source term $\sum_{m,n}\phi_{mn}J_{nm}$, the perturbative expansion stems from the generating functional of the connected correlation functions $W[J]$. It is defined by 
\begin{eqnarray}
{\cal{Z}}[J]&=&\int\prod_{m,n}d\phi_{mn}e^{-S[\phi]-\sum_{m,n}\phi_{mn}J_{nm}}=e^{W[J]}\label{generating-Z}\\
W[J]&=&\ln{\cal{Z}}(0)+W_0[J]+\ln\big(1+e^{-W_0[J]} ( e^{ -S_{int}(\frac{\delta}{\delta J})}   -1)e^{W_0[J]} \big)\label{connectgreen}\\
W_0[J]&=&\frac{1}{2}\sum_{m,n,k,l}J_{mn}P_{mn;kl}J_{kl}\label{connectfree}.
\end{eqnarray}
$S_{int}[\phi]$ denotes the interaction terms \eqn{interactmatrix} with $\Omega^2=\frac{1}{3}$ and the $J_{nm}$'s are sources. The propagator  $P_{mn;kl}$ is readily obtained from Lemma \ref{propagator-spectform} and \eqref{conservation} . Expanding the logarithm as a formal series yields all the connected diagrams while the effective action $\Gamma[\phi]$ is  obtained from $W[J]$ by Legendre transform
\begin{equation}
\Gamma[\phi]=\sum_{m,n}\phi_{mn}J_{nm}-W[J],\ \phi_{mn}=\frac{\delta W[J]}{\delta J_{nm}}\vert_{J=0}\label{effectiveaction}.
\end{equation}
The formal expansion of \eqref{connectgreen} gives rise to
\begin{eqnarray}
W^{1}[J]&=& \sum i\frac{2}{3}{\sqrt{3\mu^2}}(\delta_{m+1,n}-\delta_{m,n+1})\big((P_{lm;kl}+P_{kl;lm})(P_{nk;cd}J_{cd}+J_{ab}P_{ab;nk})\nonumber\\
&+&(P_{lm;nk}+P_{nk;lm})(P_{kl;cd}J_{cd}+J_{ab}P_{ab;kl})\nonumber\\
&+&(P_{kl;nk}+P_{nk;kl})(P_{lm;cd}J_{cd}+J_{ab}P_{ab;lm}) \big).\label{w1}
\end{eqnarray}
By usual Legendre transform applied to \eqref{w1}, $J_{mn}=\sum_{k,l}G_{nm;kl}\phi_{kl}+...$, one obtains finally ($\sigma=i\frac{2}{3}{\sqrt{3\mu^2}}$):
\begin{eqnarray}
\Gamma^1[\phi]&=&\sigma\sum_{m,n,k,l} v_{mn}\big(\delta_{mk}(P_{ll}+P_{km})\phi_{kn} +\delta_{nl}(P_{kk}+P_{nl})\phi_{ml}\nonumber\\
&+&(\delta_{m+l,n+k}P_{lk}+\delta_{n+k,m+l}P_{nm})\phi_{lk} \big)\nonumber\\
&=&\sigma\sum(2P_{ll}-P_{l,l+1}+P_{kk}+P_{k+1,k+1}-P_{k,k+1})(\phi_{k,k+1}-\phi_{k+1,k})\label{cgamma1}.
\end{eqnarray}


\end{document}